\documentclass[10pt,a4paper]{article}
\usepackage[a4paper,centering]{geometry}
\usepackage{amssymb}
\usepackage{amsthm}
\usepackage{authblk}
\usepackage{thmtools}
\usepackage{hyperref}
\input{preamble.sty}
\usepackage{stmaryrd}

\newtheorem{definition}{Definition}[section]
\newtheorem{lemma}[definition]{Lemma}
\newtheorem{proposition}[definition]{Proposition}
\newtheorem{theorem}[definition]{Theorem}

\title{Solving generic parametric linear matrix inequalities}

\author[1]{Simone Naldi}
\author[2]{Mohab Safey El Din}
\author[3]{Adrien Taylor}
\author[2]{Weijia Wang}

\affil[1]{Université de Limoges, CNRS, XLIM, Limoges, France}
\affil[2]{Sorbonne Université, CNRS, LIP6, Paris, France}
\affil[3]{Inria, École normale supérieure, PSL Research University, Paris, France}

\date{}

\begin{document}
\maketitle

\begin{abstract}
    We consider linear matrix inequalities (LMIs) $A =
        A_0+x_1A_1+\cdots+x_nA_n\succeq 0$ with the $A_i$'s being $m\times m$
    symmetric matrices, with entries in a ring
    $\gring$.
    When $\gring = \reals$, the feasibility problem
    consists in deciding whether the $x_i$'s can be instantiated to obtain a
    positive semi-definite matrix.
    When $\gring = \Q[\param_1, \ldots, \param_\nparam]$, the problem
    asks for a formula on the parameters
    $\param_1, \ldots, \param_t$, which describes
    the values of the parameters for which the specialized LMI is feasible.
    This problem can be solved using general quantifier elimination
    algorithms, with a complexity that is exponential in $n$.
    In this work, we leverage the LMI structure of the problem
    to design an algorithm that computes a formula $\Phi$ describing
    a dense subset of the feasible region of parameters, under genericity
    assumptions.
    The complexity of this algorithm
    is exponential in $n, m$ and $t$ but becomes polynomial in $n$
    when $m$ and $t$ are fixed.
    We apply the algorithm to a parametric sum-of-squares problem
    and to the convergence analyses of certain first-order optimization methods,
    which are both known to be equivalent
    to the feasibility of certain parametric LMIs, hence demonstrating its
    practical interest.
\end{abstract}

\section{Introduction}

\paragraph{Problem statement}
For a set $\mathcal{E}$ and $m\in \N$,
we denote by $\matsym_m(\mathcal{E})$
the set of $m\times m$ symmetric matrices
with entries in $\mathcal{E}$.
For $n$ and $\nparam$ in $\N$,
we consider sequences of variables $\bx = \left( x_1, \ldots, x_n \right) $
and parameters
$\by = \left( \param_1, \ldots, \param_t \right) $ and the ring $\pring =
    \Q[\by][\bx]$. The
set of polynomials in $\pring$
of degree at most $1$ in \(\bx\)
is denoted by $\linpoly{\pring}$.

A matrix $A= A_0+x_1 A_1+\cdots+x_nA_n$ in
$\matsym_m(\linpoly{\pring})$
is said to be a \emph{parametric linear matrix}.
For any $y\in \reals^\nparam$, we denote
by  $A_{y}$ the matrix obtained by specializing the parameters  $\by$
to  $y$. The linear matrix inequality $A_{y}\succeq 0$ defines the
spectrahedron $\spec(A_{y})\subset \reals^n$, i.e.,
the set of points $x$ in $\reals^n$ such that \(A_y(x)\)
is positive semi-definite. It is well-known that $\spec(A_{y})$ is a convex
semi-algebraic set, defined by $g_i(y,\cdot)\geq 0$, where
$g_i\in\pring$
is the coefficient of $\lambda^i$ in $\det(A+\lambda I_m)$
(see, e.g., \cite[Sec. 5.2]{henrion2016exact}).

\begin{problem}[Generic feasibility problem]
\label{prob:generic}
Given a parametric linear matrix
$A\in\matsym_{m}(\linpoly{\pring})$,
compute a semi-algebraic formula $\Phi$ in $\Q[\y]$,
such that $\Phi$ defines a dense subset of
the set of feasible parameters
$\mathcal{P}:=\lbrace y\in\reals^t\mid\exists x\in\reals^{n}:A(y,x)\succeq 0\rbrace$.
\end{problem}

The problem is a natural generalization of the decisional feasibility problem
for LMIs. Indeed, once the semi-algebraic formula $\Phi$ is computed,
one can decide whether a given parameter $y$ is in $\mathcal{P}$ by checking
whether $\Phi(y)$ is true.

\paragraph{Prior works}
In another perspective, the problem can be seen as a quantifier elimination
(QE) problem, which consists of eliminating the unknowns from the LMI.

Quantifier elimination over the reals has a long history starting with Tarski's
algorithm \cite{Tarski} with a non elementary recursive complexity.
Collins'cylindrical algebraic decomposition algorithm \cite{collins1975} is the
first practical algorithm. Its complexity is doubly exponential in the total
number of unknowns.

It should be noted that, here, we have a quantifier elimination problem with a
single block of variables to eliminate. The idea to exploit such a block
structure to improve the complexity of quantifier elimination originates from
\cite{Grigoriev88} and culminates in \cite{BPRQE} (see \cite[Sec.
    14.6]{basu2007algorithms} for more bibliographic notes). This yields algorithms
which are doubly exponential in the number of alternates of quantifiers and
exponential in the total number of unknowns. Still, putting into practice such
algorithms is an open problem.

A weakened variant for one-block quantifier elimination
is proposed in \cite{le2021faster} but is restricted to systems of polynomial
equations which enjoy regularity properties and do not involve inequalities.

Our study goes back to \cite{henrion2016real} which yields an \emph{exact}
algorithm for solving linear matrix inequalities, combining its principles with
those of \cite{le2021faster}.

\paragraph{Contributions}

We present an algorithm solving \Cref{prob:generic},
which can be seen as a generalization of the results in
\cite{henrion2016exact} to the parametric setting.
The algorithm first computes a finite set of polynomial systems
$(\f_i)_i$
in $\pring[\bu,\blambda]$,
where $\bu$ and $\blambda$ are auxiliary variables,
inspired by the subroutines in \cite{henrion2016exact}, used to solve
\emph{exactly} linear matrix inequalities.
The polynomial systems
$(\f_i)_i$ satisfy the following \emph{specialization property}, which we prove.
For generic values of the parameters $y\in\reals^t$,
each $\f_i(y,\cdot)=0$ has a finite number of complex solutions,
and the non-emptiness of $\spec(A_y)$ is determined by
the existence of at least one solution to a certain
$\f_i(y,\cdot)=0\land\g(y,\cdot)\geq 0$ where \(\g = (g_0, \ldots, g_m)\)
are defined above.
To extract the formula $\Phi$ characterizing this condition,
the algorithm applies to each relation $\f_i=0\land\g\geq 0$
the real root counting method described in \cite{gaillard2024},
based on Hermite's quadratic form \cite{hermite1856extrait}.

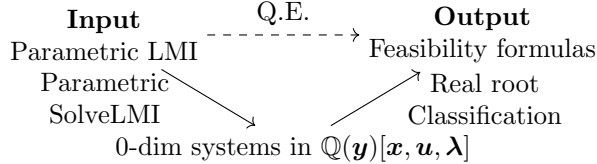
\begin{figure}[H]
    \centering
    \begin{tikzpicture}[every node/.style={align=center}]
        \node (A) at (0, 0) {\textbf{Input} \\ Parametric LMI};
        \node (B) at (5, 0) {\textbf{Output} \\ Feasibility formulas};
        \node (C) at (2.5, -1.5) {0-dim systems in $\Q(\y)[\x,\bu,\blambda]$};

        \draw[->, dashed] (A) -- (B) node[midway, above] {Q.E.};
        \draw[->] (A) -- (C) node[midway, left=4mm] {Parametric \\ SolveLMI};
        \draw[->] (C) -- (B) node[midway, right=3mm] {Real root \\ Classification};
    \end{tikzpicture}
\end{figure}

Note that for an input $\f_i$,
the algorithm in \cite{gaillard2024} is based on %
Gröbner bases computations, and, consequently,
extra genericity assumptions on the input are required
to control the complexity.

By contrast, in our case, the ideal generated by the input $\f_i$
turns out to be generically radical, but does not satisfy the genericity
assumptions used in \cite{gaillard2024}, because of the special structure of the
\(\f_i\)'s. Still, we show how to control the complexity in this case, replacing
Gröbner bases computations with the computation of %
parametric geometric resolutions \cite{schost2003computing}
of the input $\f_i$. We obtain the following complexity result.

\begin{theorem}
    \label{thm:main}
    There exists an algorithm which,
    given a generic parametric linear matrix
    $A\in\matsym_{m}(\linpoly{\pring})$
    of any fixed maximum total degree $d$ as input,
    outputs a semi-algebraic formula $\Phi$ in $\Q[\y]$
    that solves \Cref{prob:generic} for the input $A$,
    in
    \begin{equation}
        \label{eq:complexity}
        2^{O(mt)}n^{O(1)}(md)^{O(t)}(\upperbound{\delta}\upperbound{\Delta})^{O(t)}
    \end{equation}
    arithmetic operations in $\Q$, where
    \begin{align*}
        \upperbound{\delta}  \in n^{O(m^2)}\quad  \text{ and } \quad
        \upperbound{\Delta}  \in e^{O(m^2\log m)}n^{O(1)}d^{O(m^2+t)}t^{O(1)}.
    \end{align*}
\end{theorem}
Observe that for fixed \(m\) and  \(t\), the above complexity is polynomial in
\(d\) and \(n\), which is useful for the applications we have in mind (see
below).

To obtain this result, we rely on techniques, based on the critical point
method for quantifier elimination, introduced in \cite{le2021faster}. These
are restricted to the case of systems of polynomial constraints
involving equations only. We use reductions to such a situation introduced in
\cite{henrion2016exact}, but extending them to the parametric setting. As
already sketched above, we are then bound to real root classification problems,
which we solve by combining Hermite's quadratic forms as in \cite{gaillard2024}.
However, the techniques used in \cite{gaillard2024} to control the complexity
do not apply to our situation anymore because of the specific structure of the
\(\f_i\)'s. Instead, we show how to leverage this specific structure to obtain
degree bounds on algebraic sets defined by the \(\f_i\)'s (the quantities
\(\delta^\star\) and  \(\Delta^\star\) are maxima of these bounds),
so that one can control adequately the cost of the calls to
the parametric geometric resolution algorithm \cite{schost2003computing} in
replacement of Gröbner bases in the real root classification algorithm.

We also report on practical experiments,
with a first implementation of our algorithm.
Tests include topical examples of parametric LMIs which are used to
encode
convergence properties of first-order optimization algorithms
\cite{drori2014performance, taylor2017exact, taylor2017smooth}. This yields a new application area for computer algebra.
As a matter of a
fact, these parametric LMIs appear to have a moderate size \(m\) and number of
parameters  \(t\) which justify our approach.
These experiments show
that our algorithm is able to tackle such problems which are out of reach of the
state-of-the-art implementations of real quantifier elimination.

\paragraph{Plan of the paper}
Section 2 presents the basic notions of determinantal and incidence varieties
which are used to define the \(\f_i\)'s.
Section 3 presents the algorithm in \cite[Sec. 3.3]{henrion2016exact}
for solving generic LMIs.
Section 4 describes our algorithm for solving generic parametric LMIs,
whose correctness is proven in Section 5,
and whose complexity is analyzed in Section 6.
Section 7 reports on practical experiments.

\section{Preliminaries}

Throughout the paper,
we denote by $\pring=\Q[\y][\x]$,
and $\sring=\Q[\x]$,
as a special case of $\pring$ with $t=0$.

\subsection{Basic notions}

\paragraph{Critical points}

Let $\mathcal{V}\subset\C^{n}$ be a smooth and equidimensional algebraic set,
and $\phi:\C^{n}\rightarrow\C^{i}$ be a regular map,
with $i\in\interval{1}{n}$
(where $\interval{1}{n} := \lbrace 1,\ldots,n\rbrace$).
We denote by $\crit(\phi,\mathcal{V})$
the set of critical points of the restriction of $\phi$
to $\mathcal{V}$, i.e.,
the points where the Jacobian matrix $\jac(\f,\phi)$
of $(\f,\phi)$ has rank $\leq n-\dim(\mathcal{V})$,
where $\f\subset\sring$ is a polynomial system generating the ideal
$I(\mathcal{V})$ associated to \(\mathcal{V}\)
(see, e.g., \cite[Sec. 3.1]{safeyschost2003}).

\paragraph{Properness}

Let \(E\) be a topological space, $\mathcal{V}\subset E^n$,
and $\phi:\mathcal{V}\mapsto\C^{i}$ be a map,
with $i\in\interval{1}{n}$.
We say that $\phi$ is proper if for all $y\in\phi(\mathcal{V})$,
there exists a closed neighborhood $\mathcal{O}_{y}$ of $y$ such that
$\phi^{-1}(\overline{\mathcal{O}})$ is closed and bounded,
where $\overline{\mathcal{O}}$ denotes the closure of $\mathcal{O}$.

\paragraph{Change of variables}

Let $F$ be a field of characteristic $0$,
and $M\in\GL_n(\Q)$ be an invertible matrix.
For any polynomial $p\in F[\x]$, we denote by $p^M$ the polynomial
satisfying $p^M(\x)=p(M\x)$.
Analogously, for any set
$\mathcal{E}\subset F^n$,
we denote by
$\mathcal{E}^M:=\lbrace M^{-1}\x\mid\x\in\mathcal{E}\rbrace$.

\subsection{Regularity and auxiliary result}

Let \(\Variety\) be a locally closed set
of \(\complex^n\).
We say that \(\Variety\) satisfies
\(\propone\) (resp. \(\proptwo{d}\))
if the Zariski closure of  \(\Variety\) is either empty or
smooth and equidimensional (resp. of dimension  \(d\)).

We say that a polynomial system $\f\subset\pring$
satisfies assumption \(\propreg\) (resp. \(\propreg(d)\))
if $\f$ generates a radical ideal,
whose associated algebraic set $V(\f)$ satisfies \(\propone\) (resp.
\(\proptwo{d}\)).

The following lemma then serves as an equidimensional version of the Theorem
on the Dimension of Fibres \cite[Th. 1.25]{shafarevich1994basic}.

\begin{lemma}
    \label{lem:equidim}
    Let $\f\subset\Q[\y,\x]$ be a finite set of polynomials
    satisfying $\propreg{(d)}$, with $d\geq t$.
    Denote by $d^\prime=d-t$.
    Then, there exists a non-empty Zariski open set
    $\mathcal{Y}\subset\C^\nparam$ such that for all $y\in\mathcal{Y}$,
    the specialization $\f(y,\cdot)$ satisfies $\propreg{(d^\prime)}$.
\end{lemma}

\begin{proof}
    If $V(\f)$ is empty, it generates the ideal \(\langle 1\rangle\),
    as does $\f(y,\cdot)$ for all $y\in\complex^t$,
    and the lemma holds. We assume now
    \(V(\f)\neq\emptyset\).

    If $\pi_{\y}(V(\f))$ is contained in a proper Zariski closed subset
    of $\C^t$, we take $\mathcal{Y}$ to be the complement of this subset.
    Then, for all $y\in\mathcal{Y}$, $V(\f(y,\cdot))$ is empty,
    so the lemma also holds trivially.

    Hence, assume that $\pi_{\y}(V(\f))$ is Zariski dense
    in $\C^t$.
    By the Theorem on the Dimension of Fibres
    \cite[Th. 1.25]{shafarevich1994basic},
    there exists a non-empty Zariski open set
    $\mathcal{Y}^\prime$ of $\C^t$ such that for all
    $y\in\mathcal{Y}^\prime$,
    $\pi_{\y}^{-1}(y)\cap V(\f)$
    is an algebraic set of dimension $d^\prime$.

    Since \(\f\) satisfies $\propreg(d)$,
    \(\crit(\pi_{\y},V(\f))\)
    is defined by the vanishing of  \(\f\) and the
    $(n-d^\prime)$-minors of  \(\jac(\f,
    \pi_{\y})\) \cite[Lem. A.2]{din2017nearly}.
    By Sard's theorem \cite[Prop. B.2]{din2017nearly},
    $\pi_{\y}(\crit(\pi_{\y},V(\f)))$ is contained in a proper
    Zariski closed subset $Z \subset \C^\nparam$.
    Let $\mathcal{Y} = (\C^\nparam \setminus Z) \cap \mathcal{Y}^\prime$.
    Then, for all $y\in\mathcal{Y}$,
    $x\in\pi_{\x}(\pi_{\y}^{-1}(y)\cap V(\f))$,
    $\jac(\f)(y,x)$ has rank $n-d^\prime$.
    By the Jacobian criterion \cite[Th. 16.19]{eisenbud2013commutative},
    $\f(y,\cdot)$ satisfies
    $\propreg{(d^\prime)}$.
\end{proof}

\subsection{Determinantal and incidence varieties}

Let $A\in\matsym_{m}(\linpoly{\pring})$ be a parametric
linear matrix, and $r\in\interval{0}{m-1}$.
The determinantal variety $\svardet{r}\subset\C^{\nparam+n}$ of $A$
of rank $r$ is defined as
\begin{equation*}
    \mathcal{D}_{r} := \lbrace (y,x)\in\C^{\nparam+n}\mid\rank(A(y,x))\leq r\rbrace.
\end{equation*}

A desingularization of $\mathcal{D}_{r}$ in the sense of Room-Kempf
\cite{bank2010geometry} is given as follows:
let $U=(u_{i,j})_{i\in\interval{1}{m},j\in\interval{1}{m-r}}$
be a matrix of variables.
For all $\iota\subset\interval{1}{m}$ with $\card{\iota}=m-r$,
define the incidence variety
$\mathcal{V}_{r,\iota}\subset\C^{t+n+m(m-r)}$
of $A$ as
\begin{equation*}
    \mathcal{V}_{r,\iota} := \lbrace (y,x,u)\in\C^{t+n+m(m-r)}\mid A(y,x)\cdot U = 0, U_{\iota} = \Id_{m-r}\rbrace,
\end{equation*}
where $U_{\iota}$ is formed by the rows of $U$ indexed by $\iota$.

Following \cite[Lem. 3.2]{henrion2016exact}, we construct a defining
system $\f_{r,\iota}\subset\pring[\bu]$
for $\mathcal{V}_{r,\iota}$ as follows:
let $G=(g_{i,j})_{i\in\interval{1}{m},j\in\interval{1}{m-r}}$
the matrix obtained by substituting
the entries of
$U_{\iota}=\Id_{m-r}$ into $A\cdot U$,
so that $G$ is independent of $u_{i,j}$'s for $i\in\iota$.
Then, $\f_{r,\iota}$ is defined by the polynomials in $(g_{i,j})_{i\geq j}$
and $U_{\iota}-\Id_{m-r}$.

\paragraph{Properties of determinantal and incidence varieties}

The dimension of $\mathcal{D}_{r}$ and its fibres $\mathcal{D}_{r}(y,\cdot)$
are classical results.

\begin{theorem}[{\cite[Prop. 3.2]{anjos2011handbook}}]
    \label{thm:detdim}
    For all $d\in\N$, there exists a non-empty Zariski open set
    $\mathcal{A}_d\subset\matsym_{m}(\linpoly{\C[\y]_{\leq d}[\x]})$
    such that for all
    $A\in\mathcal{A}_d\cap\matsym_{m}(\linpoly{\Q[\y]_{\leq d}[\x]})$
    and
    $r\in\interval{0}{m-1}$,
    \begin{enumerate}
        \item $\mathcal{D}_{r}$ is empty or of dimension $t+n-\binom{m-r+1}{2}$;
        \item In the case where $t\geq 1$, there exists a non-empty Zariski
              open set $\mathcal{Y}_{A}\subset\C^t$ such that for all
              $y\in\mathcal{Y}_{A}\cap\Q^t$,
              $\mathcal{D}_{r}(y,\cdot)$
              is either empty or of dimension $n-\binom{m-r+1}{2}$.
    \end{enumerate}
\end{theorem}

The dimension of $\mathcal{V}_{r,\iota}$ and its fibres
$\mathcal{V}_{r,\iota}(y,\cdot)$, as well as their regularity,
are given by the following lemma, which is
a generalization of
\cite[Prop. 3.4]{henrion2016exact}.
Its proof is given in \Cref{sec:proof}.

\begin{lemma}
    \label{lem:incdim}
    For all $d\in\N$, there exists a non-empty Zariski open set
    $\mathcal{A}_d\subset\matsym_{m}(\linpoly{\C[\y]_{\leq d}[\x]})$
    such that for all
    $A\in\mathcal{A}_d\cap\matsym_{m}(\linpoly{\Q[\y]_{\leq d}[\x]})$,
    $r\in\interval{0}{m-1}$, $\iota\subset\interval{1}{m}$
    with $\card{\iota}=m-r$,
    \begin{enumerate}[left=0pt]
        \item $\f_{r,\iota}$ satisfies
              $\proptwo{t+n-\binom{m-r+1}{2}}$;
        \item when $t\geq 1$, there exists a non-empty Zariski
              open set $\mathcal{Y}_{A}\subset\C^t$ such that for all
              $y\in\mathcal{Y}_{A}\cap\Q^t$,
              $\f_{r,\iota}(y,\cdot)$ satisfies
              $\proptwo{n-\binom{m-r+1}{2}}$.
    \end{enumerate}
\end{lemma}

\section{Algorithms for LMIs}

\subsection{Principles}

Let $A\in\matsym_{m}(\sring_{\leq 1})$.
Algorithm from \cite[Sec. 3.3]{henrion2016exact} allows to decide whether
the LMI \(A\succeq 0\) is feasible.
For the ease of the reader, we recall
in the following
\cite[Th. 1.2]{henrion2016exact},
which leverages the relation between the boundary structure of
$\mathcal{S}:=\spec(A)$
and the rank stratification of $A$:

\textit{Let $A\in\matsym_{m}(\linpoly{\sring})$ be a linear matrix,
$r_{\min}=\min_{x\in\mathcal{S}}\rank(A(x))$,
and suppose that $\partial\mathcal{S}\neq\emptyset$.
Let $\mathcal{C}$ be a connected component of
$\mathcal{D}_{r_{\min}}\cap\reals^n$
such that $\mathcal{C}\cap\mathcal{S}\neq\emptyset$.
Then $\mathcal{C}\subset\mathcal{S}$.}

Consequently, solving the LMI
$A\succeq 0$ boils down to:
\begin{enumerate}[left=0pt]
    \item Pick $x\in\Q^n$,
          and decide if $x\in\mathcal{S}\setminus\partial\mathcal{S}$.
          If yes, return \texttt{true};
    \item Otherwise, $\mathcal{S}\neq\reals^n$,
          so that $\mathcal{S}\neq\emptyset$ if and only if
          $\partial\mathcal{S}\neq\emptyset$.
          We compute at least one point per connected component of
          $\mathcal{D}_{r}\cap\reals^n$,
          for all $r\in\interval{0}{m-1}$,
          and then decide if the points intersect $\mathcal{S}$.
          If yes, return \texttt{true}, otherwise return \texttt{false}.
\end{enumerate}

\subsection{Critical point methods}

For all $i\in\interval{1}{n}$,
we denote by $\pi_{i}:\C^{n}\rightarrow\C^{i}$ the projection
$x\mapsto x_{[1,i]}:=(x_1,\ldots,x_i)$,
and by abuse of notation $\pi_{0}:x\mapsto\bullet$.
We start by recalling \cite[Th. 2]{safeyschost2003}.
Let $\f\subset\sring$ and \(\Variety =
V(\f)\).

\textit{
    Assume that \(\f\) satisfies $\propreg{(d)}$.
    There exists a non-empty Zariski open set
    $\mathcal{M}\subset\GL_n(\C)$ such that for all
    $M\in\mathcal{M}\cap\GL_n(\Q)$,\\
    \textrm{(1)} The restriction of $\pi_{i-1}$ to
    $\crit(\pi_i,\mathcal{V}^M)$ is proper;\\
    \textrm{(2)} For all $i\in\interval{1}{d+1}$,
    $\crit(\pi_i,\mathcal{V}^M)$ satisfies \(\proptwo{i-1}\);
    \\
    \textrm{(3)} For all $\tau=(\tau_1,\ldots,\tau_d)\in\reals^d$,
    the set
    \begin{equation*}
        \bigcup_{i=1}^{d+1}(\crit(\pi_i,\mathcal{V}^M)\cap\pi_{i-1}^{-1}(\tau_{[1,i-1]}))
    \end{equation*}
    is finite and meets all the connected components of
    $\mathcal{V}^M\cap\reals^n$.
}
{}

We use
\cite[Sec. 5.2]{elliott2023bit} to derive Lagrange polynomial systems
defining $\crit(\pi_i,V(\f^{M}))$.

\begin{lemma}[{\cite[Sec. 5.2]{elliott2023bit}}]
    \label{lem:lagrange}
    Assume that $\f$ satisfies $\propreg{(d)}$.
    There exists a non-empty Zariski open set
    $\mathcal{M}\subset\GL_n(\C)$
    such that for all
    $M\in\mathcal{M}\cap\GL_n(\Q)$,
    it holds that
    \begin{enumerate}[left=0pt]
        \setcounter{enumi}{3}
        \item There exists a non-empty Zariski open set
              $\mathcal{T}_{M}\subset\C^n$ such that for all
              $\tau\in\mathcal{T}_{M}\cap\Q^n$,
              $i\in\interval{1}{d+1}$,
              \begin{equation*}
                  \crit(\pi_i,V(\f^{M}))\cap\pi_{i-1}^{-1}(\tau_{[1,i-1]})
                  = \pi_{\x}(\Lag(\f^{M}(\tau_{[1,i-1]},\cdot))),
              \end{equation*}
              where
              $\Lag(\f)=\f\cup(\blambda^T\jac(\f)-(1,0,\ldots,0))$
              is the Lagrange system associated with $\f$,
              with $\blambda$ a vector of variables.
    \end{enumerate}
\end{lemma}

We denote by \algopolvar an algorithm which takes as input
polynomials \(\f = (f_1, \ldots, f_s)\) and returns the Lagrange systems
defined in \Cref{lem:lagrange}. Under the assumptions of \cite[Th.
    2]{safeyschost2003}, their solutions meet all connected
components of \(V(\f)\cap\reals^n\).

\subsection{Bridging the regularity gap}

Let $A\in\matsym_{m}(\linpoly{\sring})$
be a linear matrix.
As $\mathcal{D}_{r}$ is not smooth in general,
we cannot apply \algopolvar directly to its defining system, but
only to its desingularizations $\mathcal{V}_{r,\iota}$,
where $\iota\subset\interval{1}{m}$ with $\card{\iota}=m-r$.

To compute at least one point per connected component of
$\mathcal{D}_{r}\cap\reals^n$ outside $\mathcal{D}_{r-1}$,
\cite[Prop. 3.5, 3.6]{henrion2016exact} show that it suffices to compute
critical points on the incidence varieties $\mathcal{V}_{r,\iota}$,
or the points given by the Lagrange systems of $\f_{r,\iota}$.
We restate the result as follows:

\begin{proposition}[{\cite[Prop. 3.5 and 3.6]{henrion2016exact}}]
    \label{thm:det2inc}
    Assume that
    $\mathcal{V}_{r,\iota}$ satisfies $\proptwo{n-\binom{m-r+1}{2}}$.
    Then, there exists a non-empty Zariski open set
    $\mathcal{M}\subset\GL_n(\C)$
    such that for all $M\in\mathcal{M}\cap\GL_n(\Q)$,
    for any connected component
    $\mathcal{C}\subset\mathcal{D}_{r}\cap\reals^n$,
    the following holds:
    \begin{enumerate}[left=0pt,labelwidth=0pt,itemindent=0pt]
        \item For all $i\in\interval{1}{d}$,
              $\pi_{i}(\mathcal{C}^{M})$ is closed;
              furthermore,
              for all $\tau\in\reals$ lying on the boundary of
              $\pi_{1}(\mathcal{C}^{M})$,
              $\pi_{1}^{-1}(\tau)\cap\mathcal{C}^{M}$ is finite.
        \item Let $\tau$ in the boundary of
              $\pi_{1}(\mathcal{C}^{M})$.
              Then, for all $x\in\pi_{1}^{-1}(\tau)\cap\mathcal{C}^{M}$ such that
              $x\in\mathcal{D}_{r}^{M}\setminus\mathcal{D}_{r-1}^{M}$,
              there exists $u\in\reals^{m(m-r)}$,
              $\iota\subset\interval{1}{m}$ with $\card{\iota}=m-r$,
              such that $(x,u)\in\crit(\pi_{1},\mathcal{V}_{r,\iota}^{M})$.
        \item The saturation of the ideal generated by $\Lag(\f_{r,\iota}^{M})$
              by an ideal associated with $\mathcal{D}_{r-1}^{M}$
              is radical,
              and the locally closed set
              $\mathcal{E}_{r,\iota}^{M}:=V(\Lag(\f_{r,\iota}^{M}))\cap\pi_{\x}^{-1}(\mathcal{D}_{r}^{M}\setminus\mathcal{D}_{r-1}^{M})$
              satisfies $\proptwo{\binom{m-r}{2}}$.
        \item $\pi_{\x}(\mathcal{E}_{r,\iota}^{M})$ is either empty or finite.
        \item Let $(x,u)\in\crit(\pi_{1},\mathcal{V}_{r,\iota}^{M})$ such that
              $x\in\mathcal{D}_{r}^{M}\setminus\mathcal{D}_{r-1}^{M}$.
              Then, there exists $\lambda\in\reals^{m(m-r)+\binom{m-r+1}{2}}$
              such that $(x,u,\lambda)\in\mathcal{E}_{r,\iota}^{M}$.
    \end{enumerate}
\end{proposition}

We denote by \algorealdet an algorithm which takes as input
$A\in\matsym_{m}(\linpoly{\sring})$, generates for all $r\in\interval{0}{m-1}$
and $\iota\subset\interval{1}{m}$ with $\card{\iota}=m-r$ the
systems $\f_{r,\iota}$ describing the aforementioned incidence varieties and
returns $(\f_{r,\iota,i})_{i\in\interval{1}{d+1}}$, the output of \algopolvar,
when called on these systems with their
associated rank \(r\).
The algorithm therefore consists purely of ring operations over $\sring$.

The routine \algodescartes simply takes as input \(A\) and returns the
coefficients of  \(\det(A +\lambda I_m)\) in \(\sring\) (where  \(\lambda\) is a
new variable).

Finally, the algorithm \algosolvelmi takes as input \(A\in
\matsym_{m}(\linpoly{\sring})\) and calls \algorealdet and
\algodescartes, to decide the existence of solutions to the systems returned
by \algorealdet outside the appropriate determinantal varieties
\(\svardet{r-1}\) and which make non-negative the polynomials output from
\algodescartes.

\section{Algorithms for parametric LMIs}

\subsection{Principles}

We consider now a parametric linear matrix $A\in\matsym_{m}(\linpoly{\pring})$.
We have remarked that
all operations performed by \algorealdet are ring operations over
the coefficients of the entries of \(A\). Hence,
it can be
used with entries in \(\pring\).
Hence, it returns a set of polynomial systems $F$
in $\pring[\bu,\blambda]$, which, specialized at any parameter $y\in\Q^t$,
gives us the same result if we pass the linear matrix
$A(y,\cdot)\in\matsym_{m}(\linpoly{\sring})$
to the algorithm.
Furthermore, we have the expected properties
for generic choices of $y$, as stated in the following proposition.

\begin{proposition}
    \label{prop:goal}
    For all $d\in\N$,
    there exists a non-empty Zariski open set
    $\mathcal{A}_d\subset\matsym_{m}(\linpoly{\C[\y]_{\leq d}[\x]})$
    such that for all
    $A\in\mathcal{A}_d\cap\matsym_{m}(\linpoly{\Q[\y]_{\leq d}[\x]})$,
    there exists a non-empty Zariski open set
    $\mathcal{M}_{A}\subset\GL_{n}(\C)$
    such that for all
    $M\in\mathcal{M}_{A}\cap\GL_{n}(\Q)$,
    there exists a non-empty Zariski open set
    $\mathcal{T}_{A,M}\subset\C^n$
    such that for all
    $\tau\in\mathcal{T}_{A,M}\cap\Q^n$,
    there exists a non-empty Zariski open set
    $\mathcal{Y}_{A,M,\tau}\subset\C^t$
    such that for all $y\in\mathcal{Y}_{A,M,\tau}\cap\Q^t$,
    $r\in\interval{0}{m-1}$,
    $\iota\subset\interval{1}{m}$ with $\card{\iota}=m-r$,
    $\bigcup_{r,\iota,i}\pi_{\x}(V(\f_{r,\iota,i}^{M}(y,\cdot)))$
    is finite and
    contains at least one point per connected component of each
    $\mathcal{D}_{r}^{M}(y,\cdot)$.
\end{proposition}

The proof of the proposition is technical and is given in \Cref{sec:proof}.
Let us instead focus on the consequence of the proposition.
Intuitively, the polynomial systems $\f\in F$
are zero-dimensional in $\Q(\y)[\x,\bu,\blambda]$ for generic
choices of $A$.

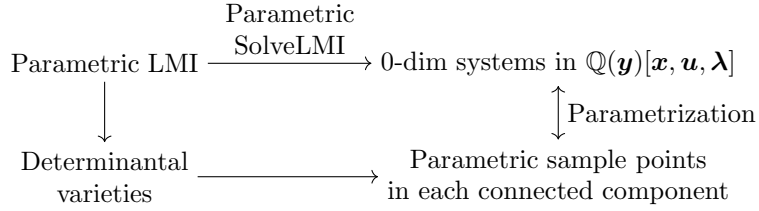
\begin{figure}[H]
    \centering
    \begin{tikzpicture}[every node/.style={align=center}]
        \node (A) at (0, 0) {Parametric LMI};
        \node (B) at (6, 0) {0-dim systems in $\Q(\y)[\x,\bu,\blambda]$};
        \node (C) at (0, -1.5) {Determinantal \\ varieties};
        \node (D) at (6, -1.5) {Parametric sample points \\ in each connected component};

        \draw[->] (A) -- (B) node[midway, above] {Parametric \\ SolveLMI};
        \draw[->] (A) -- (C) node[midway, right] {};
        \draw[->] (C) -- (D) node[midway, above] {};
        \draw[<->] (D) -- (B) node[midway, right] {Parametrization};
    \end{tikzpicture}
\end{figure}

Our goal then becomes to determine semi-algebraic formulas on $y\in\reals^t$
such that $\f(y,\cdot)=0$ and $\g(y,\cdot)\geq 0$
has a finite number of solutions.
Assume that we have the following subroutine:
\begin{itemize}[left=0pt]
    \item
          \textsc{Classification}.
          Input:
          $\f=(f_1,\ldots,f_p)\subset\pring[\bu,\blambda]$
          such that $\dim(\langle\f\rangle_{\Q(\y)})=0$,
          $\g=(g_1,\ldots,g_s)\subset\pring$.
          Output:
          a list of tuples of the form $(\Phi_{i},y_{i},r_{i})$,
          with $\Phi_{i}$ a semi-algebraic formula in $\Q[\y]$
          defining $\mathcal{T}_{i}\subset\reals^t$,
          $y_{i}\in\mathcal{T}_{i}$, $r_{i}\in\N$, such that

          \begin{itemize}
              \item For all $y\in\mathcal{T}_i$,
                    $\f(y,\cdot)=0\land\g(y,\cdot)\geq 0$
                    has $r_{i}$ real solutions,
              \item The union of $\mathcal{T}_i$ is dense in $\reals^t$.
          \end{itemize}
\end{itemize}

Then, it suffices to take the disjunction of $\Phi_i$
for which $r_i>0$ to obtain the desired semi-algebraic formulas.
We deduce \Cref{alg:paramsolvelmi} for solving the generic
feasibility problem for parametric LMIs.

\begin{algorithm}
    \caption{\textsc{ParametricSolveLMI}}
    \label{alg:paramsolvelmi}
    \begin{algorithmic}[1]
        \Require $A\in\matsym_{m}(\linpoly{\pring})$
        \Ensure A semi-algebraic formula $\Phi$ in $\Q[\y]$ defining a dense subset of the set $\lbrace y\in\reals^t\mid A(y,\cdot)\succeq 0\rbrace$.
        \State \textbf{choose} $M\in\GL_{n}(\Q)$
        \State \textbf{choose} $\tau\in\Q^n$
        \State{$\F^{M} \gets
                \textsc{RealDet}(A,M,\tau)$}\label{alg:paramsolvelmi:F}
        \State $\g^{M} \gets \textsc{ChangeVars}(\textsc{PSDMatrixCond}(A),M)$
        \State $\Phi \gets \textsc{PSDMatrixCond}(A(\cdot,0))>0$
        \ForAll{$\f_{r,\iota,i}^{M}\in\F^{M}$}
        \State $\Phi_{r,\iota,i} \gets \textsc{Classification}(\f_{r,\iota,i}^{M},\g^{M})$
        \State $\Phi \gets \Phi\lor\bigvee_{j:r_{j}>0}\Phi_{r,\iota,i,j}$
        \EndFor
        \State \Return $\Phi$
    \end{algorithmic}
\end{algorithm}

\paragraph{Proof of correctness of \Cref{alg:paramsolvelmi}}

\begin{theorem}
    Let $A\in\matsym_{m}(\linpoly{\pring})$ be a
    parametric linear matrix. Assume that $A$ and all parameters chosen in
    \Cref{alg:paramsolvelmi} satisfy the assumptions of \Cref{prop:goal}.
    Then, \Cref{alg:paramsolvelmi} is correct.
\end{theorem}

\begin{proof}
    For any $\f_{r,\iota,i}^{M}\in\F^{M}$
    given by \Cref{alg:paramsolvelmi:F} of \Cref{alg:paramsolvelmi},
    let $\mathcal{Y}_{r,\iota,i}\subset\C^{t}$ be the Zariski open set
    in \Cref{prop:goal}
    such that for all $y\in\mathcal{Y}_{r,\iota,i}\cap\Q^{t}$,
    $\Phi_{r,\iota,i}(y)$ if and only if there exists an
    $x\in\pi_{\x}(V(\f_{r,\iota,i}^{M}(y,\cdot)))$
    such that $\g^{M}(y,x)\geq 0$.
    Let $\mathcal{Y}=\cap_{r,\iota,i}\mathcal{Y}_{r,\iota,i}$,
    so that for all $y\in\mathcal{Y}$,
    $\Phi(y)$ if and only if either $A(y,0)\succeq 0$,
    or there exists an
    $x^{M}\in\bigcup_{r,\iota,i}\pi_{\x}(V(\f_{r,\iota,i}^{M}(y,\cdot)))$
    such that $\g^{M}(y,x)\geq 0$,
    i.e., \algosolvelmi returns \texttt{true}
    with input $A(y,\cdot)$.
    Finally,
    for all $y\in\mathcal{Y}\cap\mathcal{Y}_{A,M,\tau}$
    defined in \Cref{prop:goal},
    the above condition is equivalent to $\spec(A(y,\cdot))\neq\emptyset$,
    hence the correctness of \Cref{alg:paramsolvelmi}.
\end{proof}

Finally, we explain how to implement the \textsc{Classification} subroutine.
We first take \cite[Algorithm 2]{gaillard2024}, which deals with the
condition $\f(y,\cdot)=0\land\g(y,\cdot)>0$.
Then, in Line 26, we change the definition of $r_{y}$,
from the entry corresponding to $(1,\ldots,1)\in\lbrace 0,1,-1\rbrace^s$,
to the sum of entries corresponding to $\lbrace 0,1\rbrace^{s}$.
Then, $r_{y}$ is the number of real solutions
such that $\sign(\g(y,\cdot))\in\lbrace 0,1\rbrace^{s}$.

\subsection{Hermite's quadratic form}

Let $\field$ be a base field of characteristic $0$,
and $\f\subset\field[\x]$ be a polynomial system generating a
zero-dimensional ideal $\langle\f\rangle_{\field}$.
By \cite[Th. 4.86]{basu2007algorithms}, the vector space
$\mathcal{A}_{\field}:=\field[\x]/\langle\f\rangle_{\field}$
has finite dimension, denoted $\delta$.
For all $g\in\mathcal{A}_{\field}$, the Hermite's bilinear map is defined by
\begin{align*}
    \herm(\f,g) :
    \mathcal{A}_{\field}\times\mathcal{A}_{\field} & \to\field             \\
    (p,q)                                          & \mapsto \Tr(L_{gpq}),
\end{align*}
where $L_{gpq}$ is the multiplication map
$h\mapsto gpqh$ on $\mathcal{A}_{\field}$,
and $\Tr$ denotes the trace.
The associated quadratic form
$\Herm(\f,g):p\mapsto\herm(\f,g)(p,p)$
is then called Hermite's quadratic form.

\paragraph{Hermite's theorem \cite[Th. 4.102]{basu2007algorithms}}

When $\field\subset\reals$,
for all $g\in\field[\x]$,
denote by $\TaQ(g,\f)$ the Tarski-query of $g$ for $\f$,
defined as
$\TaQ(g,\f):=\card{\lbrace x\in V(\f)\mid g(x)>0\rbrace}-\card{\lbrace x\in V(\f)\mid g(x)<0\rbrace}$.

\textit{
    The following identities hold:
    \begin{itemize}[left=0pt]
        \item $\rank(\Herm(\f,g)) = \card{\lbrace x\in V(\f)\mid g(x)\neq 0\rbrace}$.
        \item $\Sign(\Herm(\f,g)) = \TaQ(g,\f)$,
    \end{itemize}
    where $\Sign(\Phi)$ denotes the signature of the quadratic form $\Phi$.
}

\paragraph{Hermite matrices}

Let $\mathcal{B}=(b_1,\ldots,b_{\delta})$ be a basis of
$\mathcal{A}_{\field}$.
For all $g\in\mathcal{A}_{\field}$,
the Hermite matrix associated with $g$
with respect to $\mathcal{B}$ is then defined as the matrix
$\mathcal{H}_{\g}=(h_{i,j})_{i,j\in\interval{1}{\delta}}\in\matsym_{\delta}(\field)$,
where $h_{i,j}=\herm(\f,g)(b_i,b_j)=\Tr(L_{gb_ib_j})$.

To compute $h_{i,j}$, let
$M_{gb_ib_j}=(m_{k,\ell})_{k,\ell\in\interval{1}{\delta}}$
be the multiplication matrix of $L_{gb_ib_j}$
with respect to $\mathcal{B}$.
By definition, $m_{k,\ell}$ is the coefficient of $b_{\ell}$
in the representative of $gb_ib_jb_k$ in $\mathcal{A}_{\field}$,
which can be obtained by a normal form computation
if $\f$ is given by a Gröbner basis.
We can then proceed by
$h_{i,j}=\Tr(M_{gb_ib_j})=\sum_{k\in\interval{1}{\delta}}m_{k,k}$.

\paragraph{Specialization property of parametric Hermite matrices}

Let $\field=\Q(\y)$,
$\f\subset\field[\x]$ generating a zero-dimensional ideal
$\langle\f\rangle_{\field}$,
and $g\in\field[\x]$.
Let $\succ_{\x}$ and $\succ_{\y}$ be monomial orders for $\x$ and $\y$,
respectively,
and $\succ$ be an elimination order which eliminates $\x$,
built with $\succ_{\x}$ and $\succ_{\y}$.
Let $\mathcal{G}$ be the reduced Gröbner basis of $\langle\f\rangle$
with respect to $\succ$,
and $\mathcal{B}$ be the monomial basis formed by the monomials of
$\field[\x]$ irreducible by $\mathcal{G}$.

The following generalizes \cite[Prop. 3.1]{gaillard2024},
by dropping the assumption that $\succ_{\x}$ and $\succ_{\y}$ are
graded reverse lexicographic orders.

\begin{theorem}
    \label{thm:hermspec}
    Let
    $\mathcal{W}_{\infty}:=\bigcup_{\gamma\in\mathcal{G}}V(\lc_{\x}(\gamma))$,
    where $\lc_{\x}(\gamma)$ is the leading coefficient of the polynomial
    $\gamma\in\field[\x]$
    with respect to $\succ_{\x}$.
    Then, for all $y\in\C^{t}\setminus\mathcal{W}_{\infty}$,
    \begin{enumerate}[left=0pt]
        \item The specialization $\mathcal{G}(y,\cdot)$ is a Gröbner basis of
              $\langle\f(y,\cdot)\rangle$ with respect to $\succ_{\x}$.
        \item Consequently,
              $\Herm(\f,g)(y,\cdot)=\Herm(\f(y,\cdot),g(y,\cdot))$.
    \end{enumerate}
\end{theorem}

\begin{proof}
    For $y\in\C^{t}\setminus\mathcal{W}_{\infty}$,
    $\lc_{\x}(g)$ does not vanish at $y$ for every $g\in\mathcal{G}$.
    By \cite[Th. 3.1]{kalkbrener1997stability},
    $\mathcal{G}(y,\cdot)$ is then a Gröbner basis of
    $\langle\f(y,\cdot)\rangle$ with respect to $\succ_{\x}$.
    We then conclude by the same proof as in
    \cite[Prop. 10]{le2022solving}.
\end{proof}

\subsection{Real root counting}
\label{sec:rrc}

Let $\field=\reals$,
$\f\subset\reals[\x]$ be a polynomial system,
and $g\in\reals[\x]$.

Denote by
$c(\sigma,\f):=\card{\lbrace x\in V(\f)\mid\sign(g(x))=\sigma\rbrace}$,
we have
\begin{equation*}
    \begin{bmatrix}
        1 & 1 & 1  \\
        0 & 1 & -1 \\
        0 & 1 & 1
    \end{bmatrix}
    \begin{bmatrix}
        c(0,\f) \\
        c(1,\f) \\
        c(-1,\f)
    \end{bmatrix}
    =
    \begin{bmatrix}
        \TaQ(g^0,\f) \\
        \TaQ(g^1,\f) \\
        \TaQ(g^2,\f)
    \end{bmatrix}.
\end{equation*}

The identity can be generalized to a tensor identity,
in the case where $\g=(g_1,\ldots,g_s)\subset\reals[\x]$.
Denote by $\SIGN(\g,\f)\subset\lbrace 0,1,-1\rbrace^{s}$
the set of realizable signs of $\g$ on $V(\f)$.
\begin{theorem}[{\cite[Prop. 10.70]{basu2007algorithms}}]
    Let $A=(\alpha_1,\ldots,\alpha_{\card{A}})\subset\lbrace 0,1,2\rbrace^{s}$,
    $\Sigma=(\sigma_1,\ldots,\sigma_{\card{\Sigma}})\subset\lbrace 0,1,-1\rbrace^{s}$,
    each in the lexicographic order.
    Suppose that $\Sigma\supset\SIGN(\g,\f)$. Then,
    \begin{equation*}
        \Mat(A,\Sigma)\cdot c(\Sigma,\f) = \TaQ(\g^{A},\f),
    \end{equation*}
    where
    \begin{itemize}[left=0pt]
        \item $\Mat(A,\Sigma)=(\sigma_{j}^{\alpha_i})_{i\in\interval{1}{\card{A}},j\in\interval{1}{\card{\Sigma}}}$,
        \item $c(\Sigma,\f)=(c({\sigma_j},\f))_{j\in\interval{1}{\card{\Sigma}}}$,
        \item $\TaQ(\g^{A},\f)=(\TaQ(\g^{\alpha_i},\f))_{i\in\interval{1}{\card{A}}}$.
    \end{itemize}
\end{theorem}

A trivial choice of $\Sigma$ is $\Sigma=\lbrace 0,1,-1\rbrace^s$,
which gives $\bm{c}$ simply by
inverting
the square matrix $\Mat(A,\Sigma)$.
By summing the coefficients corresponding to $\sigma\in\lbrace 0,1\rbrace^s$,
we obtain $a_1,\ldots,a_{\card{A}}\in\mathbb{Q}$ such that
\begin{equation}
    \label{eq:lincomb}
    \card{\lbrace x\in V(\f)\mid \g(x)\geq 0\rbrace}
    = \sum_{i=1}^{\card{A}}a_i\Sign(\Herm(\g^{\alpha_i},\f)).
\end{equation}
It suffices to extract a semi-algebraic formula corresponding to
the positivity of the linear combination of signatures.
For this, we have several options:
\begin{enumerate}[left=0pt,label=(O\arabic*)]
    \item We can directly return the positivity assertion \eqref{eq:lincomb},
          which is itself a semi-algebraic formula.
    \item Let $\mathcal{S}\subset\reals^t$ be the semi-algebraic set
          defined by the non-vanishing of the non-zero principal minor of
          maximal size
          of all Hermite matrices.
          Then, on each connected component of $\mathcal{S}$,
          the signatures of the Hermite matrices are constant.
          Thus, we can then return the \emph{realizable} combinations of signatures
          after computing one point per connected component of $\mathcal{S}$
          and examining the signatures on these points.
    \item Furthermore, the signature of a Hermite matrix is generically
          determined by the sign variation of its leading principal minors
          (see, e.g., \cite[Lem. 4.8]{gaillard2024}).
          If we add the condition of their non-vanishing in the definition
          of $\mathcal{S}$, we can refine the formula into
          a combination of signs of the leading principal minors of
          the Hermite matrices.
\end{enumerate}

When $\field=\reals(\y)$, by \Cref{thm:hermspec},
\Cref{eq:lincomb} holds for generic specializations of $y\in\reals^t$,
if we further intersect $\mathcal{S}$
with $\reals^t\setminus\mathcal{W}_{\infty}$.
This leads us to the use of \cite[Algorithm 2]{gaillard2024} which performs real
root classification for parametric zero-dimensional systems using Hermite
matrices, combined with the following subroutine:
\begin{itemize}[left=0pt]
    \item \textsc{SamplePoints}.
          Input:
          a semi-algebraic formula $\Phi$ in $\Q[\y]$.
          Output:
          a finite set of points in $\reals^t$ containing
          at least one point per connected component of the semi-algebraic set
          defined by $\Phi$. We refer to the algorithm of \cite{le2022solving}
          for an algorithm having such a specification (see also references
          therein).
\end{itemize}

\section{Proof of genericity}
\label{sec:proof}

This section is dedicated to the proof of \Cref{prop:goal},
which gives us the last piece of the puzzle to prove the correctness of
\Cref{alg:paramsolvelmi}.
The main difficulty is to show that the choices of $M\in\GL_{n}(\Q)$
and $\tau\in\Q^n$ in \Cref{prop:goal} do not depend on
the choice of parameters $y\in\reals^t$. The proof requires
a careful analysis of all the propositions and lemmas presented
in the previous sections.

\subsection{Expected dimension of incidence varieties}

\begin{proof}[Proof of \Cref{lem:incdim} (1)]
    For all $d\in\N$, denote by $n_d$ the dimension of
    $\Q[\y]_{\leq d}$.
    Observe that, as a vector space,
    $\matsym_{m}(\linpoly{\Q[\y]_{\leq d}[\x]})$
    is isomorphic to $\matsym_{m}^{(n+1)n_d}(\Q)$,
    so that we can rewrite $A$ as
    \begin{equation*}
        A(y,x)=\sum_{\substack{
                k\in\N^t, \abs{k}\leq d \\
                \ell\in\N^n, \abs{\ell}\leq 1
            }}A_{k,\ell} y^k x^\ell,
    \end{equation*}
    where
    $A_{k,\ell}=(a_{k,\ell,i,j})_{i,j\in\interval{1}{m}}\in\matsym_{m}(\Q)$.
    Denote then, by abuse of notation,
    $A\in\matsym_{m}^{(n+1)n_d}(\Q)$ the parametric linear matrix.

    Suppose w.l.o.g. that $r$ is fixed, $\iota=\interval{1}{m-r}$.
    Consider the map
    \begin{align*}
        \varphi: \C^{t+n+m(m-r)} \times\matsym_m^{(n+1)n_d}(\C)
         & \to \C^{m(m-r)+\binom{m-r+1}{2}}, \\
        (y,x,u,A)
         & \mapsto \f_{r,\iota},
    \end{align*}
    and, for $A\in\matsym_m^{(n+1)n_d}(\Q)$,
    its induced map
    $\varphi_A:\C^{t+n+m(m-r)}\to\C^{m(m-r)+\binom{m-r+1}{2}}$
    defined by $\varphi_A(y,x,u)=\varphi(y,x,u,A)=\f_{r,\iota}$,
    so that $\varphi_A^{-1}(0)=V(\f_{r,\iota})=\mathcal{V}_{r,\iota}$.

    If $\varphi^{-1}(0)=\emptyset$,
    for all $A\in\matsym_m^{(n+1)n_d}(\Q)$,
    $\varphi_A^{-1}(0)$ is empty and we are done.
    Otherwise, assume that $\varphi^{-1}(0)\neq\emptyset$.
    Note that for all $(y,x,u,A)\in\varphi^{-1}(0)$,
    $\jac(\varphi)(y,x,u,A)$ has full rank,
    given by
    the derivatives w.r.t.
    $(a_{0,0,i,j})_{i\in\interval{1}{m-r}\lor j\in\interval{1}{m-r}}$
    of the entries of $G$, and,
    the derivatives w.r.t.
    $(u_{i,j})_{i\in\iota,j\in\interval{1}{m-r}}$
    of the entries of $U_{\iota}-\Id_{m-r}$.
    Hence, $0$ is a regular value of $\varphi$.
    By Thom's weak transversality theorem \cite[Prop. B.3]{din2017nearly},
    there exists a non-empty Zariski open set
    $\mathcal{A}_d\subset\matsym_m^{(n+1)n_d}(\C)$
    such that for all $A\in\mathcal{A}_d\cap\matsym_m^{(n+1)n_d}(\Q)$,
    $0$ is a regular value of $\varphi_A$.
    By the Jacobian criterion \cite[Th. 16.19]{eisenbud2013commutative},
    $\f_{r,\iota}$ satisfies $\proptwo{t+n-\binom{m-r+1}{2}}$.
\end{proof}

\begin{proof}[Proof of \Cref{lem:incdim} (2)]
    As $\f_{r,\iota}$ satisfies $\proptwo{t+n-\binom{m-r+1}{2}}$,
    by \Cref{lem:equidim},
    there exists a non-empty Zariski open set
    $\mathcal{Y}\subset\C^t$
    such that for all $y\in\mathcal{Y}\cap\Q^t$,
    $\f_{r,\iota}(y,\cdot)$ satisfies
    $\proptwo{n-\binom{m-r+1}{2}}$.
\end{proof}

We then generalize \Cref{lem:incdim} (2), so that the choice of
$\tau\in\Q^n$ in \algorealdet becomes independent of
the choice of $y\in\Q^t$ in the following proposition.

\begin{proposition}
    Under the assumptions of \Cref{lem:incdim},
    the following additional result holds:
    \begin{enumerate}[left=0pt]
        \setcounter{enumi}{2}
        \item In the case where $t\geq 1$,
              there exists a non-empty Zariski open set
              $\mathcal{T}\subset\C^n$
              such that for all $\tau\in\mathcal{T}\cap\Q^n$,
              there exists a non-empty Zariski open set
              $\mathcal{Y}_{\tau}\subset\C^t$
              such that for all $y\in\mathcal{Y}_{\tau}\cap\Q^t$,
              for all $i\in\interval{1}{n+1}$,
              $\f_{r,\iota}(y,\tau_{\interval{1}{i-1}},\cdot)$
              satisfies $\proptwo{n-\binom{m-r+1}{2}-i}$.
    \end{enumerate}
\end{proposition}

\begin{proof}
    As $\f_{r,\iota}$ satisfies $\proptwo{t+n-\binom{m-r+1}{2}}$,
    by \Cref{lem:equidim},
    for all $i\in\interval{1}{n+1}$,
    there exists a non-empty Zariski open set
    $\mathcal{T}_i\subset\C^n$
    such that for all
    $\tau\in\mathcal{T}_i\cap\Q^n$,
    there exists a non-empty Zariski open set
    $\mathcal{Y}_{i,\tau}\subset\C^t$
    such that for all
    $y\in\mathcal{Y}_{i,\tau}\cap\Q^t$,
    $\f_{r,\iota}(y,\tau_{[1,i-1]},\cdot)$
    satisfies $\proptwo{n-\binom{m-r+1}{2}-i}$.
    We then take
    $\mathcal{T}=\cap_{i=1}^{n+1}\mathcal{T}_i$
    and
    $\mathcal{Y}_{\tau}=\cap_{i=1}^{n+1}\mathcal{Y}_{i,\tau}$.
\end{proof}

\subsection{Dancing with genericity}

We now give parametric generalizations of \Cref{thm:det2inc}. These will imply
that one can run over a fraction field the algorithms of \cite{henrion2016exact}
to reduce to real root classification problems the one of solving a parametric
LMI.

\begin{lemma}
    Assume that for all $r\in\interval{0}{m-1}$,
    $\iota\subset\interval{1}{m}$ with $\card{\iota}=m-r$,
    $\mathcal{V}_{r,\iota}$ satisfies $\proptwo{t+n-\binom{m-r+1}{2}}$.
    Then, there exists a non-empty Zariski open set
    $\mathcal{M}\subset\GL_n(\C)$
    such that for all $M\in\mathcal{M}\cap\GL_n(\Q)$,
    there exists a non-empty Zariski open set
    $\mathcal{Y}_M\subset\C^t$
    such that for all $y\in\mathcal{Y}_M\cap\Q^t$,
    for any connected component
    $\mathcal{C}\subset\mathcal{D}_{r}(y,\cdot)\cap\reals^n$,
    the results of \Cref{thm:det2inc} hold.
\end{lemma}

\begin{proof}
    The proof of \Cref{thm:det2inc} follows from
    \cite[Prop. 17, 18]{henrion2016real}.
    As the proof of \cite[Prop. 17]{henrion2016real} is purely algebraic
    and thus valid over the base field $\overline{\C(\y)}$,
    we can then invert the order of choosing $y$ and $M$
    as in \cite[Prop. 6]{le2021faster}.
\end{proof}

\begin{lemma}
    Assume that for all $r\in\interval{0}{m-1}$,
    $\iota\subset\interval{1}{m}$ with $\card{\iota}=m-r$,
    $\mathcal{V}_{r,\iota}$ satisfies $\proptwo{t+n-\binom{m-r+1}{2}}$.
    Then, there exists a non-empty Zariski open set
    $\mathcal{M}\subset\GL_n(\C)$
    such that for all $M\in\mathcal{M}\cap\GL_n(\Q)$,
    there exists a non-empty Zariski open set
    $\mathcal{Y}_M\subset\C^t$
    such that for all $y\in\mathcal{Y}_M\cap\Q^t$,
    $\iota\subset\interval{1}{m}$ with $\card{\iota}=m-r$,
    the results of \Cref{thm:det2inc} hold.
\end{lemma}

\begin{proof}
    The proof of \Cref{thm:det2inc} follows from
    \cite[Lem. 4.2 (i)]{henrion2016exact}.
    Applying the same proof shows that for generic choice of $M$,
    the system $(f,g,h)$ in \cite[Equation 4.3]{henrion2016exact}
    satisfies $Q$ in $\mathcal{O}=\lbrace (x,y,z):\rank A(x)=r\rbrace$,
    i.e., the rank of $D(f,g,h)$ is the codimension of $V((f,g,h))$
    in $V((f,g,h))\cap\mathcal{O}$.
    The choice of $y$ is then generic and follows from the Th.
    on the Dimension of Fibres \cite[Th. 1.25]{shafarevich1994basic}.
\end{proof}

\begin{proof}[Proof of \Cref{prop:goal}]
    For all $d\in\N$, let $\mathcal{A}_d$ be the
    intersection of the Zariski open sets defined in
    \Cref{thm:detdim} and \Cref{lem:incdim}.
    Then, for all
    $A\in\mathcal{A}_d\cap\matsym_{m}(\linpoly{\Q[\y]_{\leq d}[\x]})$,
    the assumptions of the above lemmas are satisfied.
    It is then straightforward to show that
    for generic choices of $M\in\GL_n(\Q)$ and $\tau\in\Q^n$,
    there exists a non-empty Zariski open set
    $\mathcal{Y}_{A,M,\tau}\subset\C^t$
    such that for all $y\in\mathcal{Y}_{A,M,\tau}\cap\Q^t$,
    the results of \Cref{thm:det2inc} hold
    for $A(y,\cdot)\in\matsym_{m}(\linpoly{\sring})$,
    and that for all $r\in\interval{0}{m-1}$,
    $\iota\subset\interval{1}{m}$ with $\card{\iota}=m-r$,
    $i\in\interval{1}{n+1}$,
    $\mathcal{D}_{r}(y,\cdot)$, $\mathcal{V}_{r,\iota}(y,\cdot)$,
    and $\f_{r,\iota,i}(y,\cdot)$
    satisfy all the results in \cite{henrion2016exact}.
    The correctness of \Cref{prop:goal} then follows from
    the correctness of \cite[Th. 3.7]{henrion2016exact}.
\end{proof}

\section{Complexity analysis}

We give now a complexity analysis of
\cite[Algorithm 2]{gaillard2024}
(and hence, by construction, of \textsc{Classification}),
without the extra genericity condition assumed in \cite{gaillard2024},
i.e., the homogeneous components of $\f$ of highest degree form
a regular sequence.
Denote $\binomsym{a}{b}=\binom{a+b}{a}$.

\subsection{Parametric geometric resolutions}

Let $\f=(f_1,\ldots,f_p)\subset\pring[\bu,\blambda]$
generating a zero-dimensional ideal
$\langle\f\rangle_{\Q(\y)}$,
and $\g=(g_1,\ldots,g_s)\in\pring$.
In the case where $\f$ generates a radical ideal, we can use parametric
geometric resolutions.

\begin{definition}
    Let $P\in\Q(\y)[\x]$ be a polynomial. We denote by
    $\denom(P)\in\Q[\y]$ the least common multiple of the
    denominators of the coefficients of $P$,
    and by $\num(P)=P\denom(P)\in\pring$.
\end{definition}

\begin{proposition}[Parametric geometric resolution, {\cite[Th. 1]{schost2003computing}}]
    \label{prop:paramres}
    Denote by $\Delta:=\deg(V(\f))$ the affine degree
    (see, e.g., \cite{heintz1983definability})
    of the algebraic set $V(\f)$, and $\delta:=\deg_{\pi}(\f)$ the
    generic cardinality of the fibers of the canonical projection
    $\pi:(\y,\x)\mapsto\y$ to $V(\f)$.
    Then,
    \begin{enumerate}[left=0pt]
        \item
              There exists a monic irreducible polynomial $q\in\Q(\y)[T]$ and
              polynomials $v_1,\ldots,v_n\in\Q(\y)[T]$,
              with
              \begin{align*}
                  \deg_{T}(q)          & = \delta,      &
                  \deg_{T}(v_i)        & \leq \delta-1,   \\
                  \deg(\denom(q))      & \leq \Delta,   &
                  \deg(\denom(v_i))    & \leq \Delta,     \\
                  \deg_{\y}(\num(q))   & \leq 2\Delta,  &
                  \deg_{\y}(\num(v_i)) & \leq 2\Delta,
              \end{align*}
              such that the parametrization
              \begin{equation*}
                  q(y,T)=0,
                  \quad
                  x_i=\frac{v_i(y,T)}{q^\prime(y,T)}
                  \quad \forall i\in\interval{1}{n}
              \end{equation*}
              gives the solutions of the system $\f(y,x)=0$ when both
              systems are specialized at all parameters
              $y=(y_1,\ldots,y_t)\in\Q^t$ such that $q,v_1,\ldots,v_n$
              are well-defined.

        \item
              There exists polynomials
              $\tilde{v}_1,\ldots,\tilde{v}_s\in\Q(\y)[T]$, with
              \begin{align*}
                  \deg_{T}(\tilde{v}_i)        & \leq \delta-1,         &
                  \deg(\denom(\tilde{v}_i))    & \leq \Delta\deg(g_i),    \\
                                               &                        &
                  \deg_{\y}(\num(\tilde{v}_i)) & \leq 2\Delta\deg(g_i),
              \end{align*}

              such that, when denoting by
              \begin{equation*}
                  \tilde{g}_i := \frac{\tilde{v}_i(y,T)}{q^\prime(y,T)}
                  \quad \forall i\in\interval{1}{s},
              \end{equation*}
              the sign of $\g(y,x)$ is the same as the sign of $\tilde{g}_i$
              for all solutions of the parametrization.
    \end{enumerate}
\end{proposition}

\begin{proof}
    \begin{enumerate}[left=0pt]
        \item
              Let $u$ be a primitive element of the extension
              $\Q(\y)\to\Q(\y)[\x]/\langle\f\rangle_{\Q(\y)}$
              with coefficients in $\Q$,
              so that $q,v_1,\ldots,v_n$,
              as well as the degree bounds,
              are given by the parametric geometric resolution
              of $\f(y,x)=0$ for $u$ \cite[Th. 1]{schost2003computing}.

        \item
              Consider the system $\f(y,x)=0$, $g_i(y,x)-\tilde{g}_i=0$,
              where $\tilde{g}_i$ is a new variable.
              As the value of $\tilde{g}_i$ is uniquely determined by other
              variables, $u$ is a primitive element of the extension
              $\Q(\y)\to\Q(\y)[\x,\tilde{g}_i]/\langle\f,g_i-\tilde{g}_i\rangle_{\Q(\y)}$,
              and $\deg_{\pi}(\f,g_i-\tilde{g}_i)=\deg_{\pi}(\f)$.
              The parametrization of $\tilde{g}_i$ is then given by the
              parametric geometric resolution of the system for $u$.
              Finally, Bézout's theorem gives
              $\deg(V(\f,g_i-\tilde{g}_i))\leq\Delta\deg(g_i)$.
              The degree bounds follow.
    \end{enumerate}
\end{proof}

We recall that the Lagrange systems
with saturation of rank defects
generically satisfy the assumptions above.
Hence, we compute the param. geometric resolution,
and pass $\f=q$ and $\g=(\tilde{g}_1,\ldots,\tilde{g}_t)$
to \textsc{Classification} instead.
We then have
$\mathcal{B}=\lbrace 1,T,\ldots,T^{\delta-1}\rbrace$,
as the unique monomial basis of $\Q(\y)[T]/\langle q\rangle_{\Q(\y)}$.
For $\alpha\in\lbrace 0,1,2\rbrace^s$,
the Hermite matrix $\mathcal{H}_{\g^{\alpha}}$ w.r.t. $\mathcal{B}$
is, using the notations of \cite{basu2007algorithms},
\begin{equation}
    \label{eq:hermdecomp}
    (\mathcal{H}_{\g^{\alpha}})_{i,j}
    = \Tr(M_{T^{i+j}\g^{\alpha}})
    = \sum_{k=1}^{\delta} \cof(T^{k},\Rem_{T}(T^{i+j+k}\g^{\alpha},q)).
\end{equation}

\paragraph{Complexity analysis}

Let $\deg(\f):=\max_{i}\deg(f_i)$ and $\deg(\g):=\max_{i}\deg(g_i)$.
By \cite[Th. 2]{schost2003computing}, there exists a probablistic algorithm
which computes the polynomials in \Cref{prop:paramres} (1) using
\begin{equation*}
    \tilde{O}\left((n\Gamma_{\f}+n^4+t^2 n)\delta\binomsym{t}{4\Delta}\right)
\end{equation*}
arithmetic operations in $\Q$,
where $\Gamma_{\f}$ is the length of a straight-line program evaluating $\f$,
and the notation $g=\tilde{O}(f)$ refers to $g=O(f\log^{\kappa}(f))$
for some $\kappa>0$.
For \Cref{prop:paramres} (2),
by simply reapplying the algorithm for each $i\in\interval{1}{s}$,
we arrive at
\begin{equation}
    \label{eq:rrccomp}
    \tilde{O}\left(\sum_{i=1}^{s}(n(\Gamma_{\f}+\Gamma_{g_i})+n^4+t^2 n)\delta\binomsym{t}{4\Delta\deg(\g)}\right).
\end{equation}

\paragraph{Bounds on $p, s$}

Let $\f=\f_{r,\iota,i}\subset\pring[\bu,\blambda]$
be one of the polynomial systems given by
\Cref{alg:paramsolvelmi:F} of \Cref{alg:paramsolvelmi},
and $\g\subset\pring$ be given by \algodescartes.
By definition, each Lagrange system of $\f_{r,\iota}$
has at most $O(n+m^2)$ polynomials,
on which we add $\binom{m}{r}$ polynomials
for the saturation of rank defects.
We then denote by $\upperbound{p}\in O(n+2^m)$ a global bound on $p$,
and similarly, $\upperbound{s}\in O(m)$.

\paragraph{Bounds on $\Gamma_{\f}$, $\Gamma_{g_i}$}

Let $\f=\f_{r,\iota,i}\subset\pring[\bu,\blambda]$
be one of the polynomial systems given by
\Cref{alg:paramsolvelmi:F} of \Cref{alg:paramsolvelmi}.
and $\g\subset\pring$ be given by \algodescartes.
To evaluate each polynomial in $\f_{r,\iota,i}$ or $\g$,
we can first evaluate the linear matrix is in
$O(m^2n\mathcal{M}_{d,t})$ arithmetic operations in $\Q$,
and then compute the evaluation of $\f$ and $\g$ purely in $\Q$.
Hence, we have $\upperbound{\Gamma_f}\in\tilde{O}(2^m n^2\mathcal{M}_{d,t})$
and $\upperbound{\Gamma_{g_i}}\in O(m^{O(1)}n\mathcal{M}_{d,t})$.

\paragraph{Bound on $\deg(\g)$}

Let $\g\subset\pring$ be the polynomial system
given by \algodescartes.
Then, we have $\deg(\g)\leq\deg(\det(\lambda I_{m}+A))=m(d+1)$,
where $d:=\deg_{\y}(A)$ is the maximum degree of the entries of $A$ in $\y$.
We then have $\upperbound{\deg(\g)}\in O(md)$.

\paragraph{Bounds on $\deg(\f), \delta$}

Let $\f=\f_{r,\iota,i}\subset\pring[\bu,\blambda]$
be one of the polynomial systems given by
\Cref{alg:paramsolvelmi:F} of \Cref{alg:paramsolvelmi}.
By \cite[Prop. 5.1]{henrion2016exact},
for all $y\in\reals^t$ such that $\f_{r,\iota,i}(y,\cdot)$
is well-defined, the specialized Lagrange system is of maximum degree $2$,
hence of total degree $d+2$,
and that $\delta\leq \upperbound{\delta}:=\mathcal{M}_{n,m(m+1)/2}^3$
when $n\leq m(m+1)/2$,
and $0$ otherwise.
Furthermore, the saturation of rank defects
consists of polynomials of total degree $r(d+1)$,
so that $\deg(\f)=\max(d+2,r(d+1))$,
i.e., $\upperbound{\deg(\f)}\in O(md)$.

\paragraph{Bound on $\Delta$}

Let $\mathcal{U} = (u_{ij})_{i,j\in\interval{1}{m}}$ be
the general symmetric matrix
with entries variables $u_{ij}$ (with $u_{ji} = u_{ij}$),
and $m^*:=m(m+1)/2$.
One has $\sing(\mathcal{D}_{r}) = \mathcal{D}_{r-1}$ and
$\mathcal{D}_{r} \subset \C^{m^*}$ is equidimensional of codimension
$c := \binom{m-r+1}{2}$ \cite{bruns2006determinantal}.
Thus, $\mathcal{D}_{r} \setminus \mathcal{D}_{r-1}$ is a constructible set
of local pure codimension $c$
and locally defined by $c$ equations $\bm{\mu} = (\mu_1,\ldots,\mu_c)$
of degree $\leq r+1$ in the entries of $\mathcal{U}$.

Consider the groups of variables $\bu = (u_{ij}, 1 \leq i \leq j \leq m)$,
$\by = (y_1, \ldots, y_t)$ and $\bx = (x_1, \ldots, x_n)$.
Let $\bv = [u_{ij} - a_{ij}(\by,\bx), 1 \leq i \leq j \leq m]$
be the list of polynomials defining the entries of the parametric LMI,
of degree $1$ with respect to $\bu$ and $\bx$,
and of degree $d$ with respect to $\by$.

Define the system $\f = [\bm{\mu}, \bv]$
consisting of $c+m^*$ polynomials in $n+t+m^*$ variables.
The Jacobian matrix of $\f$ with respect to the three variable groups above is:
\begin{equation*}
    \jac(\f) =
    \begin{bmatrix}
        \jac_{\bu}(\bm{\mu}) & 0               & 0               \\
        \Id_{m^*}            & \jac_{\by}(\bv) & \jac_{\bx}(\bv) \\
    \end{bmatrix}.
\end{equation*}
Remark that $\jac_{\bx}(\bv)$ is constant in $\bx$.
Next we introduce Lagrange multipliers
$\bm{\lambda} = (\lambda_1, \ldots, \lambda_{c+m^*})$
and consider the augmented system
$[\f, \bm{\lambda}^T \jac(\f)]$.
With respect to $(\bu,\by,\bx,\bm{\lambda})$, one has:
\begin{align*}
    \bm{\mu}                     & \rightarrow c \text{ polynomials of mdeg} \leq (r+1,0,0,0) \\
    \bv                          & \rightarrow m^* \text{ polynomials of mdeg} \leq (1,d,1,0) \\
    \bm{\lambda}^T\jac_{\bu}(\f) & \rightarrow m^* \text{ polynomials of mdeg} \leq (r,0,0,1) \\
    \bm{\lambda}^T\jac_{\by}(\f) & \rightarrow t \text{ polynomials of mdeg} \leq (0,d-1,1,1) \\
    \bm{\lambda}^T\jac_{\bx}(\f) & \rightarrow n \text{ polynomials of mdeg} \leq (0,d,0,1)
\end{align*}

The value of $\Delta$, by \cite[Prop. I.1]{din2017nearly},
is bounded from above by a multilinear Bézout bound
given by the sum of coefficients of
\begin{align*}
    ((r+1)\theta_u)^c
    (\theta_u+d\theta_y+\theta_x)^{m^*}
    (r\theta_u+\theta_\lambda)^{m^*} \\
    ((d-1)\theta_y+\theta_x+\theta_\lambda)^t
    (d \theta_y+\theta_\lambda)^n
\end{align*}
modulo
$\langle \theta_u^{m^*+1}, \theta_y^{t+1}, \theta_x^{n+1}, \theta_\lambda^{c+m^*+1} \rangle$.
Since the polynomial is homogeneous of degree $t+n+2m^*+c$,
this is also
$(r+1)^c$ times the coefficient of
$\theta_u^{m^*-c} \theta_y^{t} \theta_x^{n} \theta_\lambda^{c+m^*}$ of the
polynomial
\begin{align*}
    P := (\theta_u+d\theta_y+\theta_x)^{m^*}
    (r\theta_u+\theta_\lambda)^{m^*} \\
    ((d-1)\theta_y+\theta_x+\theta_\lambda)^t
    (d \theta_y+\theta_\lambda)^n.
\end{align*}

The multilinear bound, denoted by $\MBB(\Delta)$, is thus
\begin{align*}
      & (r+1)^c \cof(\theta_u^{m^*-c} \theta_y^{t} \theta_x^{n} \theta_\lambda^{c+m^*}, P) \\
    = & (r+1)^c \sum_{\Theta}
    d^{\alpha_y} r^{\beta_u} (d-1)^{\gamma_y} d^{\delta_y}
    \binom{m^*}{\alpha_u, \alpha_y, \alpha_x}
    \binom{m^*}{\beta_u}
    \binom{t}{\gamma_y,\gamma_x,\gamma_\lambda}
    \binom{n}{\delta_y},
\end{align*}
where $\Theta$ is the set of
$(\alpha_u,\alpha_y,\alpha_x,\beta_u,\gamma_y,\gamma_x,\gamma_{\lambda},\delta_y)$
satisfying
\begin{align}
    \label{eq:systTheta}
    \alpha_u+\alpha_y+\alpha_x              & = {m^*}            &
    \gamma_y+\gamma_x+\gamma_\lambda        & = t      \nonumber   \\
    \alpha_u+\beta_u                        & = m^*-c            &
    \alpha_y+\gamma_y+\delta_y              & = t                  \\
    \alpha_x+\gamma_x                       & = n                &
    m^*-\beta_u+\gamma_{\lambda}+n-\delta_y & = c+m^*. \nonumber
\end{align}

\begin{lemma}\label{lem_mtfixed}
    For $n > m^*+t$, we have $\MBB(\Delta)$ = 0.
    In particular, if $m$ and $t$ are fixed, the bound is $0$ for large $n$.
\end{lemma}
\begin{proof}
    The contribution of $\theta_x$ in $P$ is only given by factors
    $(\theta_u+d\theta_y+\theta_x)^{m^*}$
    and
    $((d-1)\theta_y+\theta_x+\theta_\lambda)^t$,
    thus any monomial of $P$ has degree $\leq m^*+t$ in $\theta_x$,
    so that $\Theta=\emptyset$ if $n>m^*+t$.
    If $m$ and $t$ are fixed, so does $m^*+t$,
    hence for large $n$ the bound is $0$.
\end{proof}

We remark that this is also true for the bound on $\delta$
in \cite[Prop. 5.1]{henrion2016exact}.
Let us now bound the size of $\Theta$.

\begin{lemma}\label{lem_Theta}
    $\card{\Theta} \leq \min\{(t+1)^3,(t+1)(n+1)^2\}$.
\end{lemma}
\begin{proof}
    \Cref{eq:systTheta} can be reduced to
    $L = [\alpha_u = \gamma_y+\delta_y+\gamma_x+{m^*}-t-n,
        \alpha_y = -\gamma_y-\delta_y+t,
        \alpha_x = -\gamma_x+n,
        \gamma_{\lambda} = -\gamma_y-\gamma_x+t,
        \beta_u = -\gamma_y-\delta_y-\gamma_x+t+n-c]$,
    with free $\gamma_y,\gamma_x,\delta_y$: these latter verify
    the bounds $0 \leq \gamma_y \leq t$,
    $0 \leq \gamma_x \leq \min\{t,n\}$, and
    $0 \leq \delta_{y} \leq \min\{t,n\}$.
    Thus
    $\card{\Theta} \leq (t+1)(\min\{t,n\}+1)^2$
    as claimed.
\end{proof}

\begin{lemma}\label{lem_MBB}
    Let $\upperbound{\Delta}$ be defined as
    $$
        (6r+6)^{m^*} d^{2t+m^*}
        \left(e \frac{m^*+c+t}{m^*+c}\right)^{m^*+c}
        3^t
        \min\{(t+1)^3,(t+1)(n+1)^2\}.
    $$
    Then, $\MBB(\Delta)$ is bounded from above by $\upperbound{\Delta}$.
\end{lemma}
\begin{proof}
    Recall that $\delta_y \leq t$ in $\Theta$.
    Let $D = m^*+c-(n-\delta_y)$, which is nonnegative
    (from the last equation in \eqref{eq:systTheta}, using $\beta_u \leq m^*$).
    One has $n+D = m^*+c+\delta_y \leq m^*+c+t$ and hence
    \begin{equation*}
        \binom{n}{\delta_y}
        = \binom{n}{n-\delta_y}
        = \binom{n}{m^*+c-D} \leq \binom{n+D}{m^*+c}
        \leq \binom{m^*+c+t}{m^*+c},
    \end{equation*}
    and the last binomial is bounded by
    $\left(e \frac{m^*+c+t}{m^*+c}\right)^{m^*+c}$.
    Thus from
    $\binom{m^*}{\alpha_u, \alpha_y, \alpha_x} \leq 3^{m^*}$ and
    $\binom{m^*}{\beta_u} \leq 2^{m^*}$,
    $\MBB(\Delta)$ is at most
    \begin{align*}
         &
        (r+1)^c d^{m^*} r^{m^*-c} (d-1)^t d^t
        \sum_{\Theta}
        \binom{m^*}{\alpha_u, \alpha_y, \alpha_x}
        \binom{m^*}{\beta_u}
        \binom{t}{\gamma_y,\gamma_x,\gamma_\lambda}
        \binom{n}{\delta_y}
        \\
         & \leq
        (r+1)^c d^{m^*} r^{m^*-c} (d-1)^t d^t
        \left(e \frac{m^*+c+t}{m^*+c}\right)^{m^*+c}
        3^{m^*}
        2^{m^*}
        3^t
        \card{\Theta},
    \end{align*}
    and we conclude by \Cref{lem_Theta}.
\end{proof}

\subsection{Degree bounds on Hermite matrices}

We recall that a Hermite matrix $\mathcal{H}_{\g^{\alpha}}$
has its entries given by \Cref{eq:hermdecomp}.
To compute the remainders in the equation,
as each $\tilde{g}_i$ has denominator $q^\prime$,
we need to find the inverse of $q^\prime$ modulo $q$,
or the inverse of $\num(q^\prime)$ modulo $\num(q)$.
Let $c=\Res(\num(q),\num(q^\prime))\in\Q[\y]$,
and denote by $a,b\in\Q[\y][T]$ the polynomials of degree at most $\deg_{T}(q)$
such that $a\num(q)+b\num(q^\prime)=c$.
Once $b$ is known, $c^{-1}b$ is the inverse of $\num(q^\prime)$ by $\num(q)$.

\begin{lemma}
    Let $P,Q\in\Q[\y][T]$ be polynomials.
    Denote by $M=\max(\deg_{\y}(P),\deg_{\y}(Q))$. Then,
    \begin{align*}
        \deg_{\y}(\Res(P,Q))
         & \leq M(\deg_{T}(P)+\deg_{T}(Q)),        \\
        \deg_{\y}(\Prem_{T}(P,Q))
         & \leq M\max(\deg_{T}(P)-\deg_{T}(Q)+2,1)
    \end{align*}
    Moreover, there exists $A,B\in\Q[\y][T]$,
    with $\deg_{T}(A)\leq\deg_{T}(Q)$, $\deg_{T}(B)\leq\deg_{T}(P)$,
    such that $AP+BQ=\Res(P,Q)$, and
    \begin{align*}
        \max(\deg_{\y}(A),\deg_{\y}(B)) & \leq M(\deg_{T}(P)+\deg_{T}(Q)).
    \end{align*}
    Finally, all the items above can be computed in
    $O(\deg_{T}(P)\deg_{T}(Q))$ operations in $\Q[\y]$.
\end{lemma}

\begin{proof}
    Suppose w.l.o.g. that $\deg_{T}(P)\geq\deg_{T}(Q)$.
    From \cite[Equation 4.3, Notation 8.55]{basu2007algorithms},
    using the notations of \cite{basu2007algorithms},
    we have
    \begin{align*}
        \Res(P,Q)  & = \varepsilon_{\deg_{T}(P)}\sRes_{0}(P,Q),                           \\
        \Prem(P,Q) & = \varepsilon_{\deg_{T}(P)-\deg_{T}(Q)+1}\sRes_{\deg_{T}(Q)-1}(P,Q),
    \end{align*}
    where $\varepsilon_{i}=(-1)^{i(i-1)/2}$.
    The degree bounds then follow from
    \cite[Prop. 8.71]{basu2007algorithms}.
    Moreover, \cite[Prop. 8.64]{basu2007algorithms} gives
    \begin{align*}
        \sRes_{0}(P,Q) & = \sRemU_0(P,Q)P+\sRemV_0(P,Q)Q,
    \end{align*}
    with $\sRemU_0$ and $\sRemV_0$ satisfying the same degree bounds
    as $\sRes_{0}$.
    Finally,
    all the items can be computed in \cite[Algorithm 8.81]{basu2007algorithms},
    hence the complexity.
\end{proof}

As a direct consequence,
$\deg_{T}(b)\leq\delta$,
$\deg_{\y}(b)\leq 2\Delta(2\delta-1)$,
and $\deg(c)\leq 2\Delta(2\delta-1)$.
We then factor out the denominator
from the sum in \Cref{eq:hermdecomp},
and hence deduce a degree bound on the entries of Hermite matrices
and their principal minors.

\begin{proposition}
    For $\alpha\in\lbrace 0,1,2\rbrace^s$,
    there exists $\rho\in\Q[\y]$,
    with
    \begin{equation*}
        \deg(\rho)\leq\abs{\alpha}\Delta(6\delta-3+\deg(\g))+\Delta(2\delta-2),
    \end{equation*}
    where $\abs{\alpha}:=\sum_{i=1}^{s}\alpha_i$,
    such that for all $i,j\in\interval{0}{\delta-1}$,
    $\rho M_{T^{i+j}\g^{\alpha}}$ is a matrix with entries
    in $\Q[\y]_{\leq\mathcal{D}(\abs{\alpha})}$, where
    \begin{align*}
        \mathcal{D}(\abs{\alpha}) := \
         & 2\Delta(2\delta-1)\max(1,\abs{\alpha}(\abs{\alpha}+1)(2\delta-1+\deg(\g))) \\
         & +\Delta(3\delta-3)+\abs{\alpha}\Delta
        \quad\in O(\abs{\alpha}^2(\delta+md)\delta\Delta).
    \end{align*}
\end{proposition}

\begin{proof}
    We rewrite the remainder as
    \begin{align*}
          & \Rem_{T}(T^{i+j+k}\g^\alpha, q)
        = \Rem_{T}(T^{i+j+k}\tilde{\g}^\alpha, q)                             \\
        = & \frac{\denom(q^\prime)^{\abs{\alpha}}}{c^{\abs{\alpha}}\denom(p)}
        \frac{\Prem_{T}(T^{i+j+k}\num(p),\num(q))}{\denom(q)^{\max(i+j+k+\deg_{T}(p)-\deg_{T}(q)+1,0)}},
    \end{align*}
    where $p=b^{\abs{\alpha}}\tilde{v}^{\alpha}\in\Q(\y)[T]$,
    with the degree bounds
    \begin{align*}
        \deg_{T}(p)
         & \leq \abs{\alpha}(2\delta-1),                 \\
        \deg_{\y}(\num(p))
         & \leq 2\abs{\alpha}\Delta(2\delta-1+\deg(\g)), \\
        \deg(\denom(p))
         & \leq \abs{\alpha}\Delta\deg(\g),
    \end{align*}
    so that
    \begin{align*}
             & \deg_{\y}(\Prem(T^{i+j+k}\num(p),\num(q)))                                  \\
        \leq & 2\Delta\max(\abs{\alpha}(2\delta-1+\deg(\g)),1)                             \\
             & \cdot\max(i+j+k+\abs{\alpha}(2\delta-1)-\delta+2,1)                         \\
        \leq & 2\Delta(2\delta-1)\max(1,\abs{\alpha}(\abs{\alpha}+1)(2\delta-1+\deg(\g))).
    \end{align*}

    It suffices then to take
    \begin{equation*}
        \rho=c^{\abs{\alpha}}\denom(p)\denom(q)^{\max(3\delta-2+\deg_{T}(p)-\deg_{T}(q),0)},
    \end{equation*}
    hence the degree bound on $\rho$.

    Moreover, suppose w.l.o.g. that
    $\deg_{T}(p)\geq\deg_{T}(q)$ holds.
    Then, for $k\in\interval{0}{\delta-1}$, we have
    \begin{align}
        \label{eq:cof}
        \begin{split}
              & \cof(T^k,\Rem_{T}(T^{i+j+k}\g^{\alpha},q))                  \\
            = & \rho^{-1} \cof(T^k,\Prem_{T}(T^{i+j+k}\num(p),\num(q)))     \\
              & \denom(q)^{3\delta-3-i-j-k}\denom(q^\prime)^{\abs{\alpha}},
        \end{split}
    \end{align}
    hence the degree bound on the entries of $\rho M_{T^{i+j}\g^{\alpha}}$.
\end{proof}

Hence, the numerator of each $(\mathcal{H}_{\g^{\alpha}})_{i,j}$
has degree $\leq \mathcal{D}(\abs{\alpha})$,
and the numerator of each minor of $\mathcal{H}_{\g^{\alpha}}$
has degree $\leq \delta\mathcal{D}(\abs{\alpha})$.

\subsection{Arithmetic complexity}

Recall that one can evaluate a multivariate polynomial $P$
in $k$ variables at any point within $O(\binomsym{\deg(P)}{k})$
arithmetic operations, using
\cite[Sec. 3a, 3b]{canny1989solving}.
Similarly, with $O(\binomsym{k}{\deg(P)})$ interpolation points
and the value of $P$ at these points,
one can recover $P$ within $\tilde{O}(\binomsym{k}{\deg(P)})$
arithmetic operations.
Below, we don't use
$\binomsym{\deg(P)}{k}$ and $\binomsym{k}{\deg(P)}$
interchangeably, even though they are equal.

\begin{proposition}
    \label{prop:classcomp}
    Each call to \cite[Algorithm 2]{gaillard2024} with input
    $\f=(f_1,\ldots,f_p)$, $\g=(g_1,\ldots,g_s)$ takes
    \begin{equation}
        \label{eq:classcomp}
        \tilde{O}(s\binomsym{t}{\delta\mathcal{D}(2s)}(2\delta s\beta)^{t+1}2^{3t}(\delta\mathcal{D}(2s))^{2t+1})
    \end{equation}
    arithmetic operations in $\Q$,
    where $\beta$ is the number of realizable signs of $\g$ on $V(\f)$,
    and outputs at most
    $(4\delta s\beta\mathcal{D}(2s))^t$ formulas
    that consist of $O(\delta s\beta)$ polynomials of degree
    at most $\delta\mathcal{D}(2s)$.
\end{proposition}

\begin{proof}
    By \cite[Prop. 4.11]{gaillard2024},
    each call to \textsc{SamplePoints}
    involves at most $2\delta s\beta$ polynomials,
    and hence outputs a set of at most $(4\delta^2 s\beta\mathcal{D}(2s))^t$ points in
    $\tilde{O}(\binomsym{t}{\delta\mathcal{D}(2s)}(2\delta s\beta)^{t+1}2^{3t}(\delta\mathcal{D}(2s))^{2t+1})$
    arithmetic operations in $\Q$.
    There are at most $s$ calls to \textsc{SamplePoints},
    hence the complexity in \Cref{eq:classcomp}.

    We show that
    the complexity of computing the minors (Line 9 - 13)
    and evaluating the signatures (Line 16, 25)
    is dominated by the
    complexity of computing sample points;
    as observed in \cite[Sec. 4.3]{gaillard2024},
    the other steps are simply linear algebra and negligible.

    \emph{Line 9 - 13. }
    For $\mathcal{H}_{\g^{\alpha}}$ with $\abs{\alpha}\geq 1$,
    the numerator of $(\mathcal{H}_{\g^{\alpha}})_{i,j}$,
    denoted by $P$,
    is given by \Cref{eq:hermdecomp}.
    Suppose w.l.o.g. that $\deg_{T}(p)\leq\deg_{T}(q)-1$,
    so that the coefficients are given by \Cref{eq:cof}.

    To compute $P$, we need $\binomsym{t}{\deg(P)}$ interpolation points,
    where $\deg(P)\leq \mathcal{D}(\abs{\alpha})$. \\
    For each interpolation point $y\in\Q^{t}$,
    we first compute $\num(p)(y)$, $\num(q)(y)$, $\denom(q)(y)$,
    and $\denom(q^\prime)(y)$, in
    $O(\binomsym{\deg(\num(p))}{t+1})$
    operations;
    then, for each $k\in\interval{0}{\delta-1}$,
    we compute
    $\Prem_{T}(T^{i+j+k}\num(p),\num(q))(y)$
    in $O(\abs{\alpha}\delta^2)$ operations;
    finally, we compute $P(y)$ in $O(\delta)$ operations.
    Hence, we need
    $\tilde{O}(\binomsym{t}{\mathcal{D}(\abs{\alpha})}(
        \binomsym{\deg(\num(p))}{t+1}
        +\abs{\alpha}\delta^3
        ))$
    operations.

    Consequently, computing $\mathcal{H}_{\bm{g}^{\alpha}}$
    takes the same complexity as $P$,
    as we can take the same points for all traces,
    and computing $\rho$ has a negligible cost in front of $P$.
    Also, there are at most $2\delta\beta$ principal minors to compute,
    and computing each takes at most
    \begin{equation*}
        \tilde{O}(\binomsym{t}{\delta\mathcal{D}(2s)}(
        \binomsym{4s\Delta(2\delta-1+\deg(\bm{g}))+2s(2\delta-1)}{t+1}
        +2s\delta^3
        ))
    \end{equation*}
    arithmetic operations in $\Q$,
    as it takes at most $O(\delta^3)$ operations
    for $\mathcal{H}_{\g^{\alpha}}(y)$,
    and $O(\delta^\omega)$ operations for the determinant.
    However, all the above bounds are negligible in front of
    the complexity of computing sample points, as
    \begin{align*}
            & \binomsym{4s\Delta(2\delta-1+\deg(\bm{g}))+2s(2\delta-1)}{t+1}    \\
        <   & 2(4s\Delta(2\delta-1+\deg(\bm{g}))+2s(2\delta-1))^{t+1}           \\
        \ll & (2s+1)(2\delta-1)(4s\Delta(2\delta-1+\deg(\g)))^{t+1}             \\
        <   & (\delta\mathcal{D}(2s))^{t+1} \ll (\delta\mathcal{D}(2s))^{2t+1}.
    \end{align*}

    \emph{Line 16, 25. }
    For the evaluation of the signatures,
    we recall that there are $3\delta\beta$ minors,
    whose numerators have deg. at most $\delta\mathcal{D}(2s)$,
    to evaluate at at most $(4\delta^2s\beta\mathcal{D}(2s))^t$ points.
    This is then in
    \begin{equation*}
        O(\binomsym{t}{\delta\mathcal{D}(2s)}\delta\beta(4\delta^2s\beta\mathcal{D}(2s))^t)
    \end{equation*}
    arithmetic operations in $\Q$.
\end{proof}

\begin{proof}[Proof of \Cref{thm:main}]
    For all $\f=\f_{r,\iota,i}\subset\pring[\bu,\blambda]$,
    we have $s\leq m$,
    $\mathcal{D}(2s)\leq 32(m+1)^2(\delta+md+m)\delta\Delta$,
    and $\beta\leq3^m$,
    so that the complexity of each call to \textsc{Classification},
    given by \Cref{eq:classcomp} in \Cref{prop:classcomp},
    is in
    \begin{equation*}
        \tilde{O}\left(2^{3t}m(3^m 2m\delta)^{t+1}(32m^2(\delta+md)\delta^2\Delta)^{3t+1}\right)
    \end{equation*}
    operations in $\Q$.
    As \algorealdet gives at most
    $\sum_{r=0}^{m-1}\binom{m}{r}(n-\binom{m-r+1}{2})=2^{m-3}(8n-m^2-3m+8)-n-1$
    polynomial systems,
    computing real root classifications for all $\f$
    then take the complexity
    in \Cref{eq:complexity}.

    By \cite[Th. 5.6]{henrion2016exact},
    $\algorealdet$ is in
    $n^{O(m^2)} 2^{O(m)} (n+m^2)^{O(1)}$
    arithmetic operations in $\Q[\y]$,
    which is negligible compared to the terms of \Cref{eq:complexity}
    without $\upperbound{\Delta}$,
    let alone the complexity in $\Q$ by interpolation.
    The same argument holds for \textsc{PSDMatrixCond}
    and \textsc{ChangeVars}.
    Finally,
    the complexity of
    computing each parametric geometric resolution,
    given by \Cref{eq:rrccomp},
    is in
    $2^{O(t+m)} (nt)^{O(1)} (md)^{O(t)} \upperbound{\delta}\upperbound{\Delta}^{t}$
    arithmetic operations in $\Q$,
    which is also negligible compared to the terms of \Cref{eq:complexity}.
\end{proof}

\section{Practical experiments}

We present practical experiments
with our implementation
on LMIs with background in multiple domains.
Due to their large size,
we refer to
\begin{center}
    \url{https://www-polsys.lip6.fr/~weijia/generic-param-lmi/}
\end{center}
for the exact form of the LMIs.
In the following,
we denote by \textsf{PRBtn}
the LMI for the problem \textsf{PRB}
with $t$ parameters and $n$ variables.

\paragraph{Sum-of-squares problem}

Decomposing a polynomial with real coefficients as a sum of squares (to certify
its non-negativity) is done by solving an LMI. When the coefficients of the
input polynomial depend on some parameters, one obtains a parametric LMI.
In the following, \textsf{MKN11} is extracted from a perturbation
of Motzkin's polynomial \cite{motzkin1967arithmetic},
and \textsf{RBN11} from Robinson's polynomial \cite{robinson1973some}.

\paragraph{Algorithm analyses}

We also consider LMIs extracted from the analysis of classical
first-order optimization algorithms.
In the following,
\textsf{GRD12} to \textsf{GRD23} corresponds to different choices
of parameters and specializations of a same LMI
of size $m=3$, with $2$ parameters and $4$ free variables,
generated from the gradient descent method.
Similarly,
\textsf{PPM21} and \textsf{PPM31} are generated from the proximal point method,
and \textsf{DRS32} to \textsf{DRS43} comes from
\cite[SM3.1.1, SM3.2.1]{ryu2020operator}
in the analysis of the Douglas-Rachford splitting.

\paragraph{Experimental setup}

The timings are given in hours (h.), minutes (min.) and seconds (s.),
and the computations have been performed on a server
Intel Xeon Gold 6246R with 1.5 TB of RAM.
We compute Hermite matrices using \textsf{msolve}
\cite{berthomieu2021msolve} for Gröbner basis computation with
\emph{graded reverse lexicographical} orderings
and the Maple package \textsf{Groebner}
for elimination orderings. These are used as a replacement of the parametric
geometric resolution algorithm to build our Hermite matrices (see
\cite{gaillard2024}).
The columns O1 and O2
correspond to the timings
of the first two options in \Cref{sec:rrc}.
The columns QE1 and QE2
correspond to the timings
of the quantifier elimination in Maple 2024.2 \cite{CheMM16}
and Wolfram 14.2 \cite{strzebonski2006cylindrical}.

\begin{figure}[H]
    \centering
    \scalebox{1}{
        \begin{tabular}{c|c c |c c}
                           & O1       & O2       & QE1      & QE2      \\ \hline
            \textsf{MKN11} & 5.0 s    & 1.5 s    & 5.7 s    & 0.06 s   \\
            \textsf{RBN11} & 5.0 s    & 1.6 s    & 7.1 s    & 0.04 s   \\
            \textsf{GRD12} & 1.0 s    & 3.7 s    & $\infty$ & 0.5 s    \\
            \textsf{GRD13} & 19 s     & 17 s     & $\infty$ & $\infty$ \\
            \textsf{GRD14} & $\infty$ & $\infty$ & $\infty$ & $\infty$ \\
            \textsf{GRD21} & 0.5 s    & 1.7 s    & 1.3 s    & 0.1 s    \\
            \textsf{GRD22} & 5.8 s    & 2 min    & $\infty$ & 42 min   \\
            \textsf{GRD23} & $\infty$ & $\infty$ & $\infty$ & $\infty$ \\
            \textsf{PPM21} & 0.3 s    & 0.3 s    & 0.3 s    & 0.005 s  \\
            \textsf{PPM31} & 0.3 s    & 0.4 s    & 0.4 s    & 0.007 s  \\
            \textsf{DRS32} & 2.2 s    & 8 h      & $\infty$ & $\infty$ \\
            \textsf{DRS33} & 18 min   & $\infty$ & $\infty$ & $\infty$ \\
            \textsf{DRS42} & 52 s     & $\infty$ & $\infty$ & $\infty$ \\
            \textsf{DRS43} & $\infty$ & $\infty$ & $\infty$ & $\infty$
        \end{tabular}
    }
\end{figure}
In the table, the symbol
$\infty$ indicates that the computation is unfortunately
intractable either because no result was obtained within \(48\) hours
or due to lack of memory.

When $n\geq 2$,
our implementation
succeeds for examples
where quantifier elimination in Maple 2024.2 and Wolfram 14.2
either fails due to lack of memory or does not terminate after \(48\) hours.
This illustrates the benefit of our approach on this range of problems.

\paragraph*{Acknowledgements.}

Simone Naldi is supported by the ANR Project ANR-21-CE48-0006-01 ``HYPERSPACE''.
Adrien Taylor is supported by the European Union (ERC grant CASPER 101162889).
The French government also partly funded this work under the management of
Agence Nationale de la Recherche as part of the ``France 2030'' program,
reference ANR-23-IACL-0008 ``PR[AI]RIE-PSAI''.
Mohab Safey El Din and Weijia Wang are supported by
the ANR Project ANR-22-CE91-0007 ``EAGLES''.
This collaboration started during the workshop
\emph{Conic Linear Optimization for Computer-Assisted Proofs}
in April 2022 at the Mathematisches Forschungsinstitut Oberwolfach (MFO),
organized by Etienne de Klerk, Didier Henrion, Frank Vallentin
and Angelika Wiegele.
Special thanks to Didier Henrion for his stimulating and helpful discussions
as well as his encouragements to develop this research track.
Views and opinions expressed are however those of the authors only.

\bibliographystyle{plain}
\bibliography{refs}

\begin{thebibliography}{10}

\bibitem{anjos2011handbook}
M.~F. Anjos and J.~B. Lasserre.
\newblock {\em Handbook on semidefinite, conic and polynomial optimization}.
\newblock Springer Science \& Business Media, 2011.

\bibitem{bank2010geometry}
B.~Bank, M.~Giusti, J.~Heintz, M.~Safey El~Din, and É. Schost.
\newblock On the geometry of polar varieties.
\newblock {\em Applicable Algebra in Engineering, Communication and Computing},
  21(1):33--83, 2010.

\bibitem{BPRQE}
S.~Basu, R.~Pollack, and M.-F. Roy.
\newblock On the combinatorial and algebraic complexity of quantifier
  elimination.
\newblock {\em J. ACM}, 43(6):1002–1045, November 1996.

\bibitem{basu2007algorithms}
S.~Basu, R.~Pollack, and M.-F. Roy.
\newblock {\em Algorithms in Real Algebraic Geometry}.
\newblock Algorithms and Computation in Mathematics. Springer Berlin
  Heidelberg, 2007.

\bibitem{berthomieu2021msolve}
J.~Berthomieu, C.~Eder, and M.~Safey El~Din.
\newblock Msolve: A library for solving polynomial systems.
\newblock In {\em Proceedings of the 2021 on International Symposium on
  Symbolic and Algebraic Computation}, pages 51--58, 2021.

\bibitem{bruns2006determinantal}
W.~Bruns and U.~Vetter.
\newblock {\em Determinantal rings}, volume 1327.
\newblock Springer, 2006.

\bibitem{canny1989solving}
J.~F. Canny, E.~Kaltofen, and L.~Yagati.
\newblock Solving systems of nonlinear polynomial equations faster.
\newblock In {\em Proceedings of the ACM-SIGSAM 1989 International Symposium on
  Symbolic and Algebraic Computation}, pages 121--128, 1989.

\bibitem{CheMM16}
C.~Chen and M.~{Moreno Maza}.
\newblock Quantifier elimination by cylindrical algebraic decomposition based
  on regular chains.
\newblock {\em Journal of Symbolic Computation}, 75:74--93, 2016.
\newblock Special issue on the conference ISSAC 2014: Symbolic computation and
  computer algebra.

\bibitem{collins1975}
G.~E. Collins.
\newblock Quantifier elimination for real closed fields by cylindrical
  algebraic decompostion.
\newblock In {\em Automata Theory and Formal Languages}. Springer Berlin
  Heidelberg, 1975.

\bibitem{drori2014performance}
Y.~Drori and M.~Teboulle.
\newblock Performance of first-order methods for smooth convex minimization: a
  novel approach.
\newblock {\em Mathematical Programming}, 145(1):451--482, 2014.

\bibitem{eisenbud2013commutative}
D.~Eisenbud.
\newblock {\em Commutative algebra: with a view toward algebraic geometry}.
\newblock Springer Science \& Business Media, 2013.

\bibitem{elliott2023bit}
J.~Elliott, M.~Giesbrecht, and É. Schost.
\newblock Bit complexity for computing one point in each connected component of
  a smooth real algebraic set.
\newblock {\em Journal of Symbolic Computation}, 116:72--97, 2023.

\bibitem{gaillard2024}
L.~Gaillard and M.~Safey El~Din.
\newblock Solving parameter-dependent semi-algebraic systems.
\newblock In {\em Proceedings of the 2024 International Symposium on Symbolic
  and Algebraic Computation}, pages 447--456, 2024.

\bibitem{Grigoriev88}
D.~Yu Grigor'ev.
\newblock {Complexity of deciding Tarski algebra}.
\newblock {\em Journal of Symbolic Computation}, 5(1):65--108, 1988.

\bibitem{heintz1983definability}
J.~Heintz.
\newblock Definability and fast quantifier elimination in algebraically closed
  fields.
\newblock {\em Theoretical Computer Science}, 24(3):239--277, 1983.

\bibitem{henrion2016exact}
D.~Henrion, S.~Naldi, and M.~Safey El~Din.
\newblock Exact algorithms for linear matrix inequalities.
\newblock {\em SIAM Journal on Optimization}, 26(4):2512--2539, 2016.

\bibitem{henrion2016real}
D.~Henrion, S.~Naldi, and M.~Safey El~Din.
\newblock Real root finding for determinants of linear matrices.
\newblock {\em Journal of symbolic computation}, 74:205--238, 2016.

\bibitem{hermite1856extrait}
C.~Hermite.
\newblock {Extrait d'une lettre de Mr. Ch. Hermite de Paris à Mr. Borchardt de
  Berlin sur le nombre des racines d'une équation algébrique comprises entre
  des limites données}., 1856.

\bibitem{kalkbrener1997stability}
M.~Kalkbrener.
\newblock On the stability of {G}röbner bases under specializations.
\newblock {\em Journal of Symbolic Computation}, 24(1):51--58, 1997.

\bibitem{le2021faster}
H.~P. Le and M.~Safey El~Din.
\newblock Faster one block quantifier elimination for regular polynomial
  systems of equations.
\newblock In {\em Proceedings of the 2021 International Symposium on Symbolic
  and Algebraic Computation}, pages 265--272, 2021.

\bibitem{le2022solving}
H.~P. Le and M.~Safey El~Din.
\newblock Solving parametric systems of polynomial equations over the reals
  through {H}ermite matrices.
\newblock {\em Journal of Symbolic Computation}, 112:25--61, 2022.

\bibitem{motzkin1967arithmetic}
T.~S. Motzkin.
\newblock The arithmetic-geometric inequality.
\newblock {\em Inequalities (Proc. Sympos. Wright-Patterson Air Force Base,
  Ohio, 1965)}, 205:54, 1967.

\bibitem{robinson1973some}
R.~M. Robinson.
\newblock Some definite polynomials which are not sums of squares of real
  polynomials.
\newblock {\em Selected questions of algebra and logic}, pages 264--282, 1973.

\bibitem{ryu2020operator}
E.~K. Ryu, A.~B. Taylor, C.~Bergeling, and P.~Giselsson.
\newblock Operator splitting performance estimation: tight contraction factors
  and optimal parameter selection.
\newblock {\em SIAM Journal on Optimization}, 30(3):2251--2271, 2020.

\bibitem{safeyschost2003}
M.~Safey El~Din and É. Schost.
\newblock Polar varieties and computation of one point in each connected
  component of a smooth real algebraic set.
\newblock In {\em Proceedings of the 2003 International Symposium on Symbolic
  and Algebraic Computation}, pages 224--231, 2003.

\bibitem{din2017nearly}
M.~Safey El~Din and É. Schost.
\newblock A nearly optimal algorithm for deciding connectivity queries in
  smooth and bounded real algebraic sets.
\newblock {\em Journal of the ACM (JACM)}, 63(6):1--37, 2017.

\bibitem{schost2003computing}
É. Schost.
\newblock Computing parametric geometric resolutions.
\newblock {\em Applicable Algebra in Engineering, Communication and Computing},
  13(5):349--393, 2003.

\bibitem{shafarevich1994basic}
I.~R. Shafarevich and M.~Reid.
\newblock {\em Basic algebraic geometry}.
\newblock Springer, 1994.

\bibitem{strzebonski2006cylindrical}
A.~W. Strzebo{\'n}ski.
\newblock Cylindrical algebraic decomposition using validated numerics.
\newblock {\em Journal of Symbolic Computation}, 41(9):1021--1038, 2006.

\bibitem{Tarski}
A.~Tarski.
\newblock {\em A decision method for elementary algebra and geometry}.
\newblock University of California, 1951.

\bibitem{taylor2017exact}
A.~B. Taylor, J.~M. Hendrickx, and F.~Glineur.
\newblock Exact worst-case performance of first-order methods for composite
  convex optimization.
\newblock {\em SIAM Journal on Optimization}, 27(3):1283--1313, 2017.

\bibitem{taylor2017smooth}
A.~B. Taylor, J.~M. Hendrickx, and F.~Glineur.
\newblock Smooth strongly convex interpolation and exact worst-case performance
  of first-order methods.
\newblock {\em Mathematical Programming}, 161:307--345, 2017.

\end{thebibliography}

\end{document}